\documentclass[12pt, draftclsnofoot, onecolumn]{IEEEtran}

\usepackage{graphicx}
\usepackage{amsthm}
\usepackage{epsfig}
\usepackage{latexsym}
\usepackage{amsfonts}
\usepackage{here}
\usepackage{rawfonts}
\usepackage[latin1]{inputenc}
\usepackage[T1]{fontenc}
\usepackage{calc}
\usepackage{capitalgreekitalic}
\usepackage{url}
\usepackage{enumerate}
\usepackage{color}
\usepackage[tbtags]{amsmath}
\usepackage{amssymb}
\usepackage{upref}
\usepackage{epic,eepic}
\usepackage{times}
\usepackage{dsfont}
\usepackage{comment}
\usepackage{cite}
\usepackage{color}

\renewcommand{\Pr}{\ensuremath{\operatorname{Pr}}}

\newtheorem{theorem}{\bf Theorem}

\newtheorem{definition}{\bf Definition}

\newcounter{step}
\newlength{\totlinewidth}
\newenvironment{algorithm}{%
  \rule{\linewidth}{1pt}
  \begin{list}{}%
    {\usecounter{step}%
      \settowidth{\labelwidth}{\textbf{Step 2:}}%
      \setlength{\leftmargin}{\labelwidth}%
      \setlength{\topsep}{-2pt}%
      \addtolength{\leftmargin}{\labelsep}%
      \addtolength{\leftmargin}{2mm}%
      \setlength{\rightmargin}{2mm}%
      \setlength{\totlinewidth}{\linewidth}%
      \addtolength{\totlinewidth}{\leftmargin}%
      \addtolength{\totlinewidth}{\rightmargin}%
      \setlength{\parsep}{0mm}%
      \raggedright}}%
  {\end{list}%
  \rule{\linewidth}{1pt}}
\newcounter{substep}

  {\end{list}}

\newlength{\aligntop}
\newlength{\alignbot}
 

\IEEEoverridecommandlockouts

\usepackage{algorithm}
\usepackage{algorithmic}
\usepackage{subfigure}

\begin{document}
\clearpage
\title{\LARGE Echo State Networks for Self-Organizing Resource Allocation in LTE-U with Uplink-Downlink Decoupling \vspace{-0.2cm}}
%
\author{{Mingzhe Chen\IEEEauthorrefmark{1}}, Walid Saad\IEEEauthorrefmark{2}, and Changchuan Yin\IEEEauthorrefmark{1}\vspace*{0em}\\
\authorblockA{\small \IEEEauthorrefmark{1}Beijing Laboratory of Advanced Information Network,\\ 
Beijing University of Posts and Telecommunications, Beijing, China 100876\\
Email: \protect\url{chenmingzhe@bupt.edu.cn} and \protect\url{ccyin@ieee.org.} \\
\IEEEauthorrefmark{2}Wireless@VT, Bradley Department of Electrical and Computer Engineering, Virginia Tech, Blacksburg, VA, USA, Email: \protect\url{walids@vt.edu}\\
\vspace{-0.8cm}
 }

%
%
%
%

\maketitle
\thispagestyle{empty}
\vspace{0cm}
\begin{abstract}
Uplink-downlink decoupling in which users can be associated to different base stations in the uplink and downlink of heterogeneous small cell networks (SCNs) has attracted significant attention recently. However, most existing works focus on simple association mechanisms in LTE SCNs that operate
only in the licensed band. In contrast, in this paper, the problem of resource allocation with uplink-downlink decoupling is studied for an
SCN that incorporates LTE in the unlicensed band (LTE-U). Here, the users can access both licensed and unlicensed bands while being associated to different base stations. This problem is formulated as a noncooperative game that incorporates user association, spectrum allocation, and load balancing.
To solve this problem, a distributed algorithm based on the machine learning framework of \emph{echo state networks} {(ESNs)} is proposed using which the small base stations autonomously choose their optimal bands allocation strategies while having only limited information on the network's and users' states. It is shown that the proposed algorithm converges to a stationary mixed-strategy distribution which constitutes a mixed strategy Nash equilibrium for the studied game. Simulation results show that the proposed approach yields significant gain, in terms of the sum-rate of the $50$th percentile of users, that reaches up to {$167\%$} compared to {a Q-learning algorithm}. The results also show that ESN significantly provides a considerable reduction of information exchange for the wireless network.
\end{abstract}
\vspace{0cm}

\newpage
\section{Introduction}
\label{sec:intro}
The recent surge in wireless services has led to significant changes in existing cellular systems \cite{1}. In particular, the next-generation of cellular systems will be based on small cell networks (SCNs) that rely on low-cost, low-power small base stations (SBSs). The ultra dense nature of SCNs coupled with the transmit power disparity between SBSs, constitute a key motivation for the use of uplink-downlink decoupling techniques \cite{2,3} in which users can associate to different SBSs in the uplink and downlink, respectively. Such techniques have become recently very popular, particularly with the emergence of uplink-centric applications \cite{LearningHow} such as machine-to-machine communications {and} social networks. 

The existing literature has studied a number of problems related to uplink-downlink decoupling {such as in \cite{2} and \cite{3}}. In \cite{2}, the authors delineate the main benefits of decoupling the uplink and downlink, and propose an optimal decoupling association strategy that maximizes data rate. The work in \cite{3} investigates the throughput and outage gains of uplink-downlink decoupling using a simulation approach. Despite the promising results, these existing works are restricted to performance analysis and centralized optimization approaches that may not scale well in a dense and heterogeneous SCN. Moreover, these existing works are restricted to classical LTE networks in which the devices and SBSs can access only a single, licensed band. 

Recently, there has been {a} significant interest in studying how LTE-based SCNs can operate in the unlicensed band (LTE-U){\cite{5,6,7,2016CULTEGuan,Ko2016A,8,9,10,11,12,AMultiGame,HolisticSmallCell,zhang2015hierarchical}}. LTE-U presents many challenges in terms of {spectrum allocation, user association, interference management, and 
co-existence \cite{5,6,7,2016CULTEGuan}}.   
In \cite{5}, optimal resource allocation algorithms are proposed for both dual band femtocell and integrated femto-WiFi networks. The authors in \cite{6} develop a hybrid method to perform both traffic offloading and resource sharing in an LTE-U scenario {using a co-existence mechanism based on optimizing the duty cycle of the system.} In \cite{7}, {the authors analyze, using stochastic geometry, the performance of LTE-U with continuous transmission, duty cycle, and listen-before-talk (LBT) co-existence mechanisms}. {The authors in \cite{2016CULTEGuan} propose a cognitive co-existence scheme to enable spectrum sharing between LTE-U and WiFi networks.} However, most existing works on LTE-U {\cite{5,6,7,2016CULTEGuan,Ko2016A,8,9,10,11,12,AMultiGame,HolisticSmallCell,zhang2015hierarchical}} have focused on performance analysis and resource allocation {in LTE-U systems, under conventional association methods}. Indeed, none of these works analyzed the potential of uplink-downlink decoupling in LTE-U. LTE-U provides an ideal setting to perform uplink-downlink {decoupling since} the possibility of uplink-downlink decoupling exists not only between base stations but also between the licensed and unlicensed bands.   

More recently, reinforcement learning (RL) techniques have gained significant attention for developing distributed approaches for resource allocation in LTE and heterogeneous SCNs \cite{13,14,15,16,17,18,19,20,21}. In \cite{14,15,16,17} and \cite{21}, channel selection, network selection, and interference management were addressed using the framework of Q-learning and smoothed best response. In \cite{18,19,20}, regret-based learning approaches are developed to address the problem of interference management, dynamic clustering, and SBSs' on/off. However, none of the existing works on RL \cite{13,14,15,16,17,18,19,20,21} have focused on the LTE-U network and downlink-uplink decoupling. Moreover, most existing algorithms \cite{13,14,15,16,17,18,19,20,21}, require agents to obtain the other agents' value functions \cite{22} and state information, which is not practical for scenarios in which agents are distributed. In contrast, here, our goal is to develop a novel and efficient multi-agent RL algorithm based on recurrent neural networks (RNNs) \cite{23,24,25,26,27,28} whose advantage is that they can store the state {information of the users} and have no need to share value functions between agents.

Since RNNs have the ability to retain state over time, because of their recurrent connections, they are promising candidates for compactly storing moments of series of observations. \emph{Echo state networks} {(ESNs)}, an emerging RNN framework \cite{23,24,25,26,27,28}, are a promising candidate for wireless network modeling, as they are known to be relatively easy to train. Existing literature has studied a number of problems related to ESNs \cite{23,24,25,26}. In \cite{23,24,25}, ESNs are proposed to model reward function, characterize wireless channels, and equalize the non-linear satellite communication channel. The authors in \cite{26}, prove the convergence of ESNs-RL algorithm {in the context} of Markov decision problems and {develop suitable algorithms} to settle {problems}. However, most existing works on ESNs \cite{23,24,25,26} have {mainly} focused on {problems in the operations research literature with little work that investigated the use of ESN in a wireless networking environment}. {Indeed}, none of these works exploited the use of ESNs for resource allocation problem in a wireless network. 

The main contribution of this paper is to develop a novel, self-organizing framework to optimize resource allocation with uplink-downlink decoupling in an LTE-U system. 
We formulate the problem as a noncooperative game in which the players are the SBSs and the macrocell base station (MBS). Each player seeks to find an optimal spectrum allocation scheme to optimize a utility function that captures the sum-rate in terms of downlink and uplink, and balances the licensed and unlicensed {spectra} between users. To solve this resource allocation game, we propose a self-organizing algorithm based on the powerful framework of ESNs \cite{26,27,28}. {Here, we note that the use of a self-organizing approach to schedule the LTE-U resource can reduce the coordination between base stations (BSs) which, in future SCNs, can be significantly limited by the backhaul capacity. Moreover, next-generation cellular networks will be dense and, as such, centralized control can be difficult to implement which has motivated the use of self-organizing approaches for resource allocation such as \cite{14,17,18} and \cite{21}.} Unlike previous studies {such as \cite{5} and \cite{6}}, which rely on the coordination among SBSs and on the knowledge of the entire users' state {information}, the proposed approach requires minimum information to learn the mixed strategy Nash equilibrium {(NE)}. The use of ESNs enables the LTE-U SCN to quickly learn its resource allocation parameters without requiring significant training data. 
The proposed algorithm enables \emph{dual-mode} SBSs to autonomously learn and decide on the allocation of the licensed and unlicensed bands in the uplink and downlink to each user depending on the network environment. Moreover, we show that the proposed ESN algorithm can converge to a mixed strategy {NE}.
To our best knowledge, \emph{this is the first work that exploits the framework of ESNs to optimize resource allocation with uplink-downlink decoupling in LTE-U systems}. 
Simulation results show that the proposed approach yields a performance improvement, in terms of the sum-rate of the $50$th percentile of users, reaching up to {$167\%$} compared to {a Q-learning approach}.

The rest of this paper is organized as follows. The system model is described in Section \uppercase\expandafter{\romannumeral2}. The ESN-based resource allocation algorithm is proposed in Section \uppercase\expandafter{\romannumeral3}. In Section \uppercase\expandafter{\romannumeral4}, numerical simulation results are presented and analyzed. Finally, conclusions are drawn in Section \uppercase\expandafter{\romannumeral5}.

\vspace{-0cm}
\section{System Model and Problem Formulation}
\label{sec:SM}
Consider the downlink and uplink of an SCN that encompasses LTE-U, WiFi access points {(WAPs)}, and a macrocell network. Here, the macrocell tier operates using only the licensed band. The MBS is located at the center of a geographical area. Within this area, we consider a set $\mathcal{N} = \{1,2,\ldots,N_s\}$ of \emph{dual-mode} SBSs that are able to access both the licensed and unlicensed bands. Co-located with this cellular network is a WiFi network that consists of $W$ {WAPs}. In addition, we consider a set $\mathcal{U} = \left\{ {1,2, \cdots ,U} \right\}$ of $U$ LTE-U users which are distributed uniformly over the area of interest. All users can access different SBSs as well as the MBS for transmitting in the downlink or uplink. For this system, we consider an FDD mode for LTE on the licensed band, which splits the licensed band into equal spectrum bands for the downlink and uplink. {We consider an LTE system that uses a TDD duplexing mode along with a duty cycle mechanism to manage the co-existence over the unlicensed band as done in \cite{LTEin} and \cite{LTEU}}. {Using the duty cycle method, the SBSs will use a discontinuous, duty-cycle transmission pattern so as to guarantee the transmission rate of WiFi users. Compared to LBT, in which the SBSs transmits data during continuous time slots, the duty cycle method enables the SBSs to use several discontinuous (not necessarily consecutive) time slots. This discontinuous, duty-cycle transmission method is similar to the popular idea of an almost-blank subframe \cite{7}. Under this method, the time slots on the unlicensed band will be divided between LTE-U and WiFi users. In particular, LTE-U transmits for a fraction $\vartheta$ of time and will be muted for $1-\vartheta$ time which is allocated to the WiFi transmission. 
The LTE-U transmission duty cycle $\eta$ consists of discontinuous time slots which are adaptively adjusted based on the WiFi data rate requirement. A static muting pattern for LTE-U enables all the SBSs transmit (or mute) either synchronously or asynchronously. If the SBSs are muted synchronously, they transmit or mute at the same time. In contrast, if the SBSs are muted asynchronously, the neighboring SBSs can adopt a shifted version of the muting pattern. In our model, we consider the static synchronous muting pattern such as in \cite{6} and \cite{7}. We assume that the WiFi network transmits data during $L_w$ time slots as the LTE network transmits one LTE frame. Consequently, during these $L_w$ time slots, $L_l$ discontinuous time slots are allocated to the LTE-U network and can be normalized as $L={{{L_l}} \mathord{\left/
 {\vphantom {{{L_l}} {{L_w}}}} \right.
 \kern-\nulldelimiterspace} {{L_w}}}$ and $1-L$ fraction of time slots on the unlicensed band is used for the transmission of the WiFi users {as shown in Fig. \ref{f2}.} 
For LTE-U operating {over} the unlicensed band, TDD offers the flexibility to adjust the resource allocation between the downlink and uplink. The {WAPs} will transmit using a standard carrier sense multiple access with collision avoidance (CSMA/CA) protocol {and} its corresponding {request-to-send/clear-to-send (RTS/CTS)} access mechanism. 
\begin{figure}[!t]
  \begin{center}
   \vspace{0cm}
    \includegraphics[width=8cm]{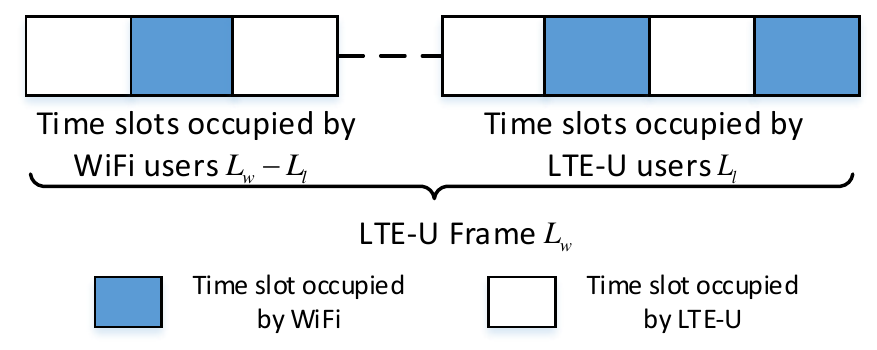}
    \vspace{-0.3cm}
  {\caption{\label{f2}Illustrative example on the time slot distribution between WiFi and LTE-U users.}}
  \end{center}\vspace{-0.8cm}
\end{figure}
\subsection{LTE data rate analysis}
{Hereinafter, we use the term BS to refer to either an SBS or the MBS and we denote by $\mathcal{B}$ the set of BSs} {and $\mathcal{B}_l$ be the set of the {BSs} on the licensed band.} During the connection period, we denote by $c_{lij}^\textrm{DL}$ the downlink capacity and $c_{lji}^\textrm{UL}$ the uplink capacity of user $i$ that is associated with BS $j$ on the licensed band. Thus, the overall long-term downlink and uplink rates of LTE-U user $i$ on the licensed band are given by:
 \begin{equation}\label{eq:Rldl}
  {R_{lij}^\textrm{DL}} = {d_{ij}}{c_{lij}^\textrm{DL}}, 
\end{equation}
 \begin{equation}
{R_{lji}^\textrm{UL}} = {v_{ji}}{c_{lji}^\textrm{UL}}, 
\end{equation}
where
 \begin{equation*}
 c_{lij}^\textrm{DL}\!=\! F_l^\textrm{DL}{\log _2}\left(\!1\!+\! {\frac{{{P_j}{h_{ij}}}}{{\sum\limits_{k \in \mathcal{B}_l,k \ne j} {{P_k}{h_{ik}}}  \!+\! {\sigma ^2}}}} \!\right)\!\!,
\end{equation*}
 \begin{equation*} 
 c_{lji}^\textrm{UL}\! =\! F_l^\textrm{UL}{\log _2}\!\left(\!\!1\!+\!\! {\frac{{{P_u}{h_{ij}}}}{{\sum\limits_{k \in \mathcal{U}, k \ne i} {{P_u}{h_{kj}}}  \!+\! {\sigma ^2}}}}\! \right)\!,
 \end{equation*}
 $d_{ij}$ and $v_{ji}$ are the fraction of the downlink and uplink licensed bands allocated from SBS $j$ to user $i$, respectively, 
$F_l^\textrm{DL}$ and $F_l^\textrm{UL}$ denote, respectively, the downlink and uplink bandwidths on the licensed band, {$P_{j} \in \{P_M,P_P \}$ is the transmit power of BS $j$, where $P_M$ and $P_P$, represent, respectively, the transmit power of the MBS and each SBS}, $P_{u}$ is the transmit power of LTE-U users, ${h_{ij}}$ is the channel gain between user $i$ and BS $j$, and ${\sigma ^2}$ is the power of the Gaussian noise. 

Similarly, the downlink and uplink rates of user $i$ that is transmitting over the unlicensed band are given by:
\begin{equation}
R_{uij}^\textrm{DL} = {\kappa _{ij}}c_{uij}^\textrm{DL},
\end{equation}
\begin{equation}\label{eq:Ruul}
R_{uji}^\textrm{UL} = {\tau_{ji}}c_{uji}^\textrm{UL},
\end{equation}
where
 \begin{equation*}
 c_{uij}^\textrm{DL}={{LF_u}{{\log }_2}\left( {1 + \frac{{{P_P}{h_{ij}}}}{{\sum\limits_{k \in \mathcal{B}_u, k \ne j} {{P_P}{h_{ik}} + {\sigma ^2}} }}} \right)} ,
\end{equation*}

 \begin{equation*}
c_{uji}^\textrm{UL}\! =\! LF_u{\log _2}\!\left(\!\!1\!+\!\! {\frac{{{P_u}{h_{ij}}}}{{\sum\limits_{k \in \mathcal{U}, k \ne i} {{P_u}{h_{kj}}}  \!+\! {\sigma ^2}}}}\! \right)\!,
\end{equation*}
${\kappa _{ij}}$ and ${\tau_{ji}}$ denote, respectively, the downlink and uplink time slots during which user $i$ transmits on the unlicensed band. Note that, the SBSs adopt a TDD mode on the unlicensed band and the LTE-U users on the uplink and downlink share the time slots of the unlicensed band. 
$F_u$ denotes the bandwidth of the unlicensed band, $\mathcal{B}_u$ denotes the SBSs on the unlicensed band, and $L$ $\in [0,1]$ is {the fraction of} time slots during which the LTE SCN uses the unlicensed band. {Here, we note that the interference generated over the unlicensed band comes from other SBSs.} 
\subsection{WiFi data rate analysis}
{We consider a WiFi network at its saturation capacity with binary slotted exponential backoff mechanism such as in \cite{6} and \cite{29}. In this model, each WiFi user will immediately have a packet available for transmission, after the completion of each successful transmission.} {This WiFi model can be applied to the WiFi network based on the different protocols (e.g. 802.11n) such as in \cite{6} and \cite{Small, Admission, ANover}.} {Since in one LTE frame time slot, the WiFi network will occupy $L_w-L_l$ WiFi time slots while the SBSs will use the $L_l$ WiFi time slots, WiFi contention model in \cite{29} can indeed be considered here.} The saturation capacity of $N_w$ users sharing the same unlicensed band can be expressed by \cite{29}: 
 \begin{equation}
 R\left( {{N_w}} \right) = \frac{{{P_\textrm{tr}}{P_s}E\left[ S \right]}}{{\left( {1 - {P_\textrm{tr}}} \right){T_\sigma } + {P_\textrm{tr}}{P_s}{T_s} + {P_\textrm{tr}}\left( {1 - {P_s}} \right){T_c}}},
\end{equation}
where ${P_\textrm{tr}} = 1 - {\left( {1 - \varsigma } \right)^{{N_w}}}$, ${P_\textrm{tr}}$ is the probability that there is at least one transmission in a time slot and $\varsigma$ is the transmission probability of each user. ${P_s} =\left. {N_w}\varsigma {\left( {1 - \varsigma } \right)^{{N_w} - 1}}\small/{P_\textrm{tr}}\right.$,
is the successful transmission on the channel. {Here, $(1-\varsigma)^{N_w}$ is the probability that all WiFi users are not using the unlicensed band but are in the backoff stage or detection stage. The transmission probability $\varsigma$ depends on the conditional collision probability and the backoff windows.} $T_s$ is the average time that the channel is sensed busy because of a successful transmission, $T_c$ is the average time that the channel is sensed busy by each station during a collision, $T_\sigma$ is the duration of an empty slot time, and $E\left[ S \right]$ is the average packet size. {Therefore, ${{P_{\textrm{tr}}}{P_s}E\left[ S \right]}$ is the average amount of payload information successfully transmitted during a given time slot and ${{\left( {1 - {P_{\textrm{tr}}}} \right){T_\sigma } + {P_{\textrm{tr}}}{P_s}{T_s} + {P_{\textrm{tr}}}\left( {1 - {P_s}} \right){T_c}}}$ is the average length of a time slot. Clearly, we can see that $R\left( {{N_w}} \right)$ is actually the average number of bits that were successfully transmitted during the air time occupied by $N_w$ WiFi users. Hence, we can see that $R(N_w)/N_w$ is indeed the average number of bits that were transmitted by each WiFi user. Given this definition, we can then use the average amount of information transmitted in a time slot as a measure of the average rate and, as such, we can compare $R(N_w)/N_w$ with the average rate requirement of a given user since this average rate requirement will effectively be equivalent to the average amount of information transmitted in a time slot.
Note that the number of the users that are associated with the WAPs is considered to be pre-determined and will not change during the operation of the proposed algorithm \cite{2016CULTEGuan}}.

In our model, the WiFi network adopts conventional distributed coordination function (DCF) access and RTS/CTS access mechanisms. Thus, $T_c$ and $T_s$ are given by \cite{29}:
\begin{equation}
 \begin{aligned}
{T_s} = &{{RTS} \mathord{\left/
 {\vphantom {{RTS} C}} \right.
 \kern-\nulldelimiterspace} C} + {{CTS} \mathord{\left/
 {\vphantom {{CTS} C}} \right.
 \kern-\nulldelimiterspace} C} + {{\left( {H + E\left[ P \right]} \right)} \mathord{\left/
 {\vphantom {{\left( {H + E\left[ P \right]} \right)} C}} \right.
 \kern-\nulldelimiterspace} C} \\
 &+ {{ACK} \mathord{\left/
 {\vphantom {{ACK} C}} \right.
 \kern-\nulldelimiterspace} C} + 3SIFS + DIFS + 4\delta ,
 \end{aligned}
 \end{equation}
 \begin{equation}
 {T_c} = {{RTS} \mathord{\left/
 {\vphantom {{RTS} C}} \right.
 \kern-\nulldelimiterspace} C} + DIFS + \delta,
 \end{equation}
where $T_s$ denotes the average time the channel is sensed busy because of a successful transmission, $T_c$ is the average time the channel is sensed busy by each station during a collision, {$H$ is the packet header}, $C$ is the channel bit rate, {$\delta$ is the propagation delay, $ACK$, $DIFS$, $RTS$, and $CTS$ represent, respectively, the time of the acknowledgement, distributed inter-frame space, RTS and CTS.}

{Based on the duty cycle mechanism,  {$\left(1-L\right)$ fraction of time slots} will be allocated to the WiFi users. We assume that the rate requirement of WiFi user is $R_w$. In order to satisfy the WiFi user rate requirement $R_w$, the {fraction} of time slots $1-L$ is given as:
\begin{equation}\label{eq:RNw}
\frac{{R\left( {{N_w}} \right)\left( {1 - L} \right)}}{{{N_w}}} \ge {R_w},
\end{equation}}   
where {${{R({N_w})} \mathord{\left/
{\vphantom {{R({N_w})} {{N_w}}}} \right.
 \kern-\nulldelimiterspace} {{N_w}}}$ is the rate for each WiFi user. From (8), the fraction of time slots on the unlicensed band that are allocated to the LTE-U network is given by ${{L \le 1 - {N_w}{R_w}} \mathord{\left/
 {\vphantom {{L \le 1 - {N_w}{R_w}} {R({N_w})}}} \right.
 \kern-\nulldelimiterspace} {R({N_w})}}$}. {From (\ref{eq:RNw}), we can see that the duty cycle is actually decided by the WiFi user average data rate requirement and the number of the WiFi users in the network.} 

\subsection{Problem formulation}
Given this system model, our goal is to develop an effective spectrum allocation scheme with uplink-downlink decoupling that can allocate the appropriate bandwidth on the licensed band and time slots on the unlicensed band to each user, simultaneously. {The decoupling essentially implies that the downlink and uplink of each user can be associated with different SBSs as well as the LTE MBS. Indeed, we consider the effect of WiFi users on the LTE-U transmissions but we do not consider the users' associations with the WAPs. For example, a user can be associated in the uplink to an LTE-U SBS and in the downlink to the LTE macrocell, or a user can be associated in the uplink to LTE-U SBS 1 and in the downlink to LTE-U SBS 2 that is different from SBS 1.} However, the rate of each {BS} depends not only on its own choice of the allocation action but also on remaining {BSs'} actions. In this regard, we formulate a noncooperative game $\mathcal{G} = \left[ {\mathcal{B},\left\{ {{\mathcal{A}_n}} \right\}_{n \in \mathcal{B}},\left\{ {{u_n}} \right\}}_{n \in \mathcal{B}} \right]$ in which the set of BSs $\mathcal{B}$ are the players including SBSs and the MBS and $u_n$ is the utility function for BS $n$. Each player $n$ has a set ${\mathcal{A}_n} = \left\{ {\boldsymbol{a}_{n,1}, \ldots, \boldsymbol{a}_{n,{\left| {{\mathcal{A}_n}} \right|}}} \right\}$  of actions where $\left| {{\mathcal{A}_n}} \right|$ is the total number of actions. For an SBS $n$, each action ${\boldsymbol{a}_n} = \left( {{\boldsymbol{d}_n},{\boldsymbol{v}_n},{\boldsymbol{\rho}_n}} \right)$, is composed of: (i) the downlink licensed bandwidth allocation ${ \boldsymbol{d}_n} = \left[  {d_{n,1}, \ldots, d_{n,{K_n}}} \right]$, where $K_n$ is the number of all users in the coverage area $\mathcal{L}_b$ of SBS $n$, (ii) the uplink licensed bandwidth allocation ${ \boldsymbol{v}_n} = \left[  {v_{n,1}, \ldots v_{n,{{K_n}}}} \right]$, and, (iii) the time slots allocation on the unlicensed band ${ \boldsymbol{\rho }_n} = \left[  {\kappa_{1,n}, \ldots, \kappa_{{K_n},n}, \tau_{n,1}, \ldots, \tau_{n,{{K_n}}}} \right]$. ${\boldsymbol{d}_n}$, ${\boldsymbol{v}_n}$ and ${\boldsymbol{\rho}_n}$ must satisfy:
\begin{equation}\label{eq:c1}
\sum\nolimits_{j=1}^{K_n}\! {d_{n,j}} \le 1,\:\: \sum\nolimits_{j=1}^{K_n} \!{v_{n,j}} \le 1,
\end{equation}
\begin{equation}\label{eq:c2}
\sum\nolimits_{j = 1}^{K_n}\! \left( {\kappa_{n,j}+\tau_{j,n}} \right) \le 1,
\end{equation}
\begin{equation}\label{eq:c3}
{d_{j,n},v_{n,j},\kappa_{n,j}, \tau_{j,n}} \in \mathcal{Z},
\end{equation}
where $\mathcal{Z} = \left\{ {{1}, \ldots ,{Z}} \right\}$ is a finite set of $Z$ level fractions of spectrum. For example, one can separate the fractions of spectrum into $Z=10$ equally spaced intervals. Then, we will have $\mathcal{Z} = \left\{ {{0}, {0.1}, \ldots ,{0.9}, {1}} \right\}$. For the MBS, each action ${\boldsymbol{a}_m} = \left( {\boldsymbol{d}_m},{\boldsymbol{v}_m}\right)$ is composed of its downlink licensed bandwidth allocation ${ \boldsymbol{d}_m}$ and uplink licensed bandwidth allocation ${ \boldsymbol{v}_m}$. ${\boldsymbol{a}} = \left( {{\boldsymbol{a}_{1}},{\boldsymbol{a}_{2}}, \ldots ,{\boldsymbol{a}_{N_b}}} \right) \in \mathcal{A}$, represents the action profile of all players where $N_b=N_s+1$, expresses the number of BSs including one MBS and $N_s$ SBSs, and $ \mathcal{A} = \prod\nolimits_{n \in N_b} {{\mathcal{A}_{n}}}$. 

To maximize the downlink and uplink rates simultaneously while maintaining load balancing, for each SBS $n$, the utility function needs to capture both the sum data rate and load balancing. Here, load balancing implies that each {SBS} will balance its spectrum allocation between users, while taking into account their capacity. Therefore, we define a utility function $u_{n}$ for SBS $n$ is:
\begin{equation}\label{eq:un}
\begin{split}
u_{n} \left(\boldsymbol{a}_{n}, {\boldsymbol{a}_{-n}}\right) =&\sum\limits_{n = 1}^{K_j} \underbrace{\log _2\left(1+d_{nj}c_{lnj}^{\textrm{DL}}+\eta \kappa_{nj}c_{unj}^{\textrm{DL}}\right)}_{\text{Downlink rate}} \! \\&+\sum\limits_{n = 1}^{K_j} \underbrace{\log _2\left(1+v_{jn}c_{ljn}^{\textrm{UL}}+\tau_{jn}c_{ujn}^{\textrm{UL}}\right)}_{\text{Uplink rate}},
\end{split}
 \end{equation}
where $\boldsymbol{a}_{-n}$ denotes the action profile of all the BSs other than SBS $n$ and $\eta$ is a scaling factor that adjusts the number of users to use the unlicensed band in the downlink due to the high outage probability on the unlicensed band. For example, for $\eta=0.5$, the downlink rates of users on the unlicensed band decreased by half, which, in turn, will decrease the value of utility function, which will then require the SBSs to change their other spectrum allocation schemes in a way to achieve a higher value of utility function. {In fact, (\ref{eq:un}) captures the downlink and uplink sum-rate over the licensed and unlicensed bands based on (\ref{eq:Rldl})-(\ref{eq:Ruul}). {From (\ref{eq:un}), we can see that the downlink rate has no relationship with the uplink rate, which means that the uplink and downlink of each user can be associated with different SBSs and/or MBS.} In (12), the logarithmic function is used to balance the load between the users.} Since the MBS {can only allocate the licensed spectrum to the users}, the utility function $u_{m}$ for the MBS can be expressed by: 
\begin{equation}\label{eq:um}
\begin{split}
u_{m}\! \left(\boldsymbol{a}_{m}, {\boldsymbol{a}_{-m}}\right)\! \!=\!\!&\sum\limits_{n = 1}^{K_m} {\log _2\!\left(1\!+\!d_{mj}c_{lmj}^{\textrm{DL}}\right)}\!\!+\!\!\!\sum\limits_{n = 1}^{K_m} {\log _2\!\left(1\!+\!v_{jm}c_{ljm}^{\textrm{UL}}\right)}.
\end{split}
 \end{equation}
Note that, hereinafter, we use the utility function $u_n$ to refer to either the utility function $u_n$ of SBS $n$ or the utility function $u_m$ of the MBS.
Given a finite set $\mathcal{A}$, $\Delta\left( \mathcal{A} \right)$ represents the set of all probability distributions over the elements of {$\mathcal{A}$}. Let ${\boldsymbol{\pi} _n}= \left[ {{\pi _{n,\boldsymbol{a}_{1}}}, \ldots ,{\pi _{n,\boldsymbol{a}_{\left| {{A_n}} \right|}}}} \right]$ be a probability distribution using which BS $n$ selects a given action from $\mathcal{A}_n$. Consequently, ${\pi _{n,\boldsymbol{a}_{i}}} = \Pr \left( {{\boldsymbol{a}_n} = {\boldsymbol{a}_{n,i}}} \right)$ is BS $n$'s mixed strategy where $\boldsymbol{a}_n$ is the action that BS $n$ adopts. Then, the expected reward that BS $n$ adopts the spectrum allocation scheme $i$ given the mixed strategy   long-term performance metric can be written as follows:
\begin{equation}\label{eq:eun}
{\mathbb{E}}\left[ {{u_n}\left( {{{\boldsymbol{a}_{n,i}}}} \right)} \right] = \sum\limits_{{\boldsymbol{a}_{ - n}} \in {\mathcal{A}_{ - n}}} {{u_n}\left( {{{\boldsymbol{a}_{n,i}}},{\boldsymbol{a}_{ - n}}} \right){\pi _{ - n, \boldsymbol{a}_{ - n}}}},
\end{equation}
where ${\pi _{ - n,{\boldsymbol{a}_{ - n}}}} \!\!=\!\sum\nolimits_{{{\boldsymbol{a}_{n,i}}} \in { \mathcal{A}_n}} {\!\pi \!\left( {{{\boldsymbol{a}_{n,i}}},{\boldsymbol{a}_{ - n}}} \right)} $ denotes the marginal probabilities distribution over the action set of BS $n$.   

 \section{Echo State Networks for Self-Organizing Resource Allocation} 
Given the proposed wireless model {in} Section \uppercase\expandafter{\romannumeral2}, our next goal is to solve the proposed resource allocation game. To solve the game, our goal is to find the mixed-strategy {NE}. The mixed NE is formally defined as follows:

\begin{definition}\emph{(Mixed Nash equilibrium): A mixed strategy profile ${\boldsymbol{\pi} ^*} = \left( {\boldsymbol{\pi} _1^*, \ldots ,\boldsymbol{\pi} _{N_b}^*} \right) = \left( {\boldsymbol{\pi} _n^*,\boldsymbol{\pi}_{ - n}^*} \right)$ is a mixed strategy Nash equilibrium if, $\forall n \in \mathcal{B}$ and $\forall \boldsymbol{\pi} _n \in \Delta \left( {{\mathcal{A}_n}} \right)$, it holds that:
\begin{equation}\label{eq:mNE}
{\tilde u_n}\left( {\boldsymbol{\pi} _n^*,\boldsymbol{\pi} _{ - n}^*} \right) \ge {\tilde u_n}\left( {{\boldsymbol{\pi} _n},\boldsymbol{\pi} _{ - n}^*} \right),
\end{equation}
where 
\begin{equation}
{\tilde u_n}\left( {{\boldsymbol{\pi} _n},{\boldsymbol{\pi} _{ - n}}} \right) =\sum\limits_{{\boldsymbol{a}} \in {\mathcal{A}}} {{u_n}\left( {{\boldsymbol{a}}} \right)\prod\limits_{n \in \mathcal{B}} {{\pi _{n,{a_n}}}}} 
\end{equation}
is the expected utility of BS $n$ when selecting the mixed strategy $\boldsymbol{\pi} _n$.  }
\end{definition}

For our game, the mixed NE represents each BS maximizes its data rate and balances the licensed and unlicensed {bands} between {the} users. 
However, in a dense SCN, each SBS may not know the entire users' states information including interference, location, and path loss, which makes it challenging to solve the proposed game in the presence of limited information. To find the mixed NE, we need to develop a novel learning algorithm based on the powerful framework of \emph{echo state networks}.

ESNs are {a new type} of recurrent neural networks \cite{26,27,28} that can be easy to train and can track the state of a network over time. Learning algorithms based on ESNs can learn to mimic a target system with arbitrary accuracy and can automatically adapt spectrum allocation to the change of network states. Consequently, ESNs are promising candidates for solving wireless resource allocation games in SCNs whereby each SBS can use an ESN approach to simulate the users' states, estimate the value of the aggregate utility function, and find the mixed strategy NE of the game. Here, finding the mixed strategy NE of the proposed game refers to the process of allocating appropriate bandwidth on the licensed band and time slots on the unlicensed band to each user, simultaneously. 
Compared to traditional RL approaches such as Q-learning \cite{30}, an approach based on ESNs can quickly learn the resource allocation parameter without requiring significant training data and it has the ability to adapt the optimal spectrum allocation scheme over time, due to the use of recurrent neural network concepts. 

In order to find the mixed strategy NE of the proposed game, we begin by describing how to use ESNs for learning resource allocation parameters. Then, we propose an ESN-based approach to find the mixed strategy NE. Finally, we prove the convergence of the proposed learning algorithm with different learning rules. 


\begin{figure*}[htbp]
  \centering
  \includegraphics[width=0.8\linewidth]{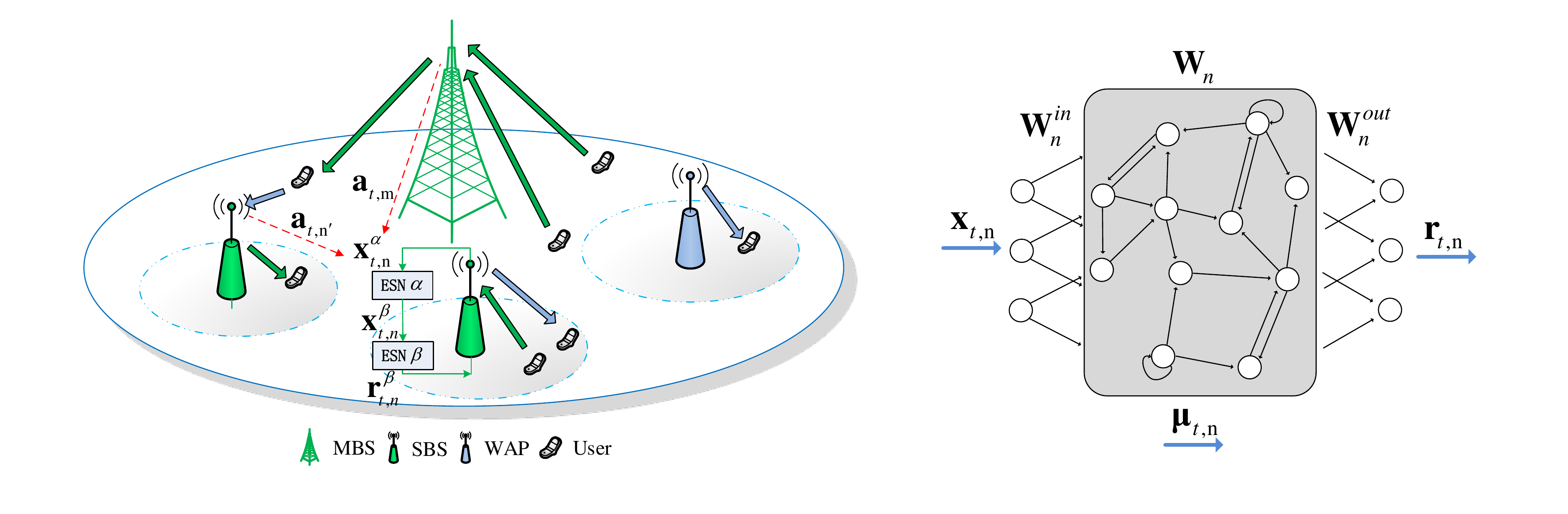}
  \caption{An SCN using ESNs with uplink-downlink decoupling in LTE-U. The left hand shows the proposed ESNs learning algorithm running in SBS. The right hand shows the architecture of {an ESN}.}\label{fig.2}
\vspace{-0.6cm}
\end{figure*}

\subsection{{Formulation based on Echo State Networks}}
As shown in Fig. \ref{fig.2}, two ESNs are used in our proposed algorithm. Each such algorithm consists of four components: a) agents, b) inputs, c) actions, and d) reward function. The agents are the players who using ESNs learning algorithm in our proposed game. The inputs, actions, and reward functions are used to find the mixed strategy NE. In our proposed algorithm, the first ESN is proposed to approximate the utility function of BSs and the second ESN is used to choose the optimal spectrum allocation action. Thus, they have different inputs and reward functions but the same agents and actions. {Here, the reward function of the first ESN captures the gain of each spectrum allocation scheme and the reward function of the second ESN captures the expected gain of each spectrum allocation scheme. These reward functions will be based on the utility functions in (\ref{eq:un})-(\ref{eq:eun})}. Hereinafter, we use {ESN} $\alpha$ and {ESN} $\beta$ to refer to the first {ESN} and the second {ESN}, respectively. The ESN models are thus defined as follows: 

$\bullet$ \textbf{Agents}: The agents are the BSs in $\mathcal{B}$.  

$\bullet$ \textbf{Inputs:} $\boldsymbol{x}_{t,n}^\alpha\!=\!\left[ {\boldsymbol{a}_{1}\!\left(t\right)\!, \cdots\!, \boldsymbol{a}_{n-1}\!\left(t\right), \boldsymbol{a}_{n+1}\!\left(t\right), \cdots\!, \boldsymbol{a}_{N_b}\!\left(t\right) } \right]^{\mathrm{T}}$ which represents the state of network at time $t$, $\boldsymbol{a}_{i}\left(t\right)$ is the spectrum allocation scheme that BS $i$ adopts at time $t$.  

$\boldsymbol{x}_{t,n}^\beta=\left[ {x_{t,n1}^\textrm{DL}, \cdots , x_{t,K_nj}^\textrm{DL}, x_{t,n1}^\textrm{UL}, \cdots , x_{t,K_nj}^\textrm{UL}} \right]^{\mathrm{T}}$ which represents the association of the LTE-U users of BS $n$ at time $t$ ( i.e., $x_{t,ni}^\textrm{DL}=1$ when user $i$ connects to BS $n$ in the downlink). Actually, at each time, all the users send the connection requests to all BSs. Therefore, inputs only have the request state and stable state, i.e., $\left[1, \dots, 1\right]$ and $\left[1, \dots, 0\right]$. 

$\bullet$ \textbf{Actions:} {Each BS $n$ can only use} one spectrum allocation action at time $t$, i.e., ${{\boldsymbol{a}_n}\left( t \right) = {\boldsymbol{a}_{n,i}}}$, where ${\boldsymbol{a}_{n,i}} \in {\mathcal{A}_n}$. Therefore, ${\boldsymbol{a}_{n,i}}$ is specified as follows:
\begin{equation}
\boldsymbol{a}_{n,i}= {\left[ {\begin{array}{*{20}{c}}
{d_{1,n}^{i} \cdots d_{{K_n,n}}^{i}, v_{n,1}^{i} \cdots v_{{n,K_n}}^{i}}\\
{\kappa_{1,n}^{i} \cdots \kappa_{{K_n,n}}^{i}, \tau_{n,1}^{i} \cdots \tau_{{n,K_n}}^{i}}
\end{array}} \right]^\mathrm{T}},
\end{equation}
where $d_{j,n}^{i}$ is the fraction of the downlink licensed band that SBS $n$ allocates to user $j$ adopting the spectrum allocation scheme $i$, $v_{n,j}^{i}$ denotes the fraction of the uplink licensed band that SBS $n$ allocates to user $j$ adopting spectrum allocation scheme $i$, $\kappa_{j,n}^{i}$ and $\tau_{n,j}^{i}$ respectively, denote the fractions of the unlicensed band that SBS $n$ allocates to user $j$ in the downlink and uplink. We consider the case in which each user can access the licensed and unlicensed bands at the same time. The licensed bandwidths and time slots on the unlicensed band must satisfy (\ref{eq:c1})-(\ref{eq:c3}). Thus, ${\boldsymbol{a}_{n,i}}$ represents SBS $n$ adopting spectrum allocation action $i$ to the users. {Since the MBS can only allocate the licensed spectrum to the users, the action of the MBS in ${\mathcal{A}_m}$} is given by:
\begin{equation}
\boldsymbol{a}_{m,i}= {\left[ {\begin{array}{*{20}{c}}
{d_{1,m}^{i} \cdots d_{{K_m,m}}^{i}, v_{m,1}^{i} \cdots v_{{m,K_m}}^{i}}\\
\end{array}} \right]^\mathrm{T}}.
\end{equation}


$\bullet$ \textbf{Reward function:} In our model, the vector of the reward functions are the functions to which ESNs approximate. The reward function of {ESN} $\alpha$ is used to store the reward of spectrum allocation action. Therefore, the reward function is given by: 
\begin{equation}\nonumber
\boldsymbol{r}_{t,n}^\alpha\! \left(\boldsymbol{x}_{t,n}^\alpha, {\mathcal{A}_n}\right)\!=\!\!\left[r_{t,n}^{\alpha,1} \!\!\left(\boldsymbol{x}_{t,n}^\alpha, {\boldsymbol{a}_{n,1}}\right)\!,\! \cdots\!, r_{t,n}^{\alpha, {\left| {{\mathcal{A}_n}} \right|}}\!\!\left(\boldsymbol{x}_{t,n}^\alpha, {\boldsymbol{a}_{n,{{\left| {{\mathcal{A}_n}} \right|}}}}\right)\!\right]^{\mathrm{T}}\!\!\!,
\end{equation}
where $\boldsymbol{r}_{t,n}^\alpha \left(\boldsymbol{x}_{t,n}^\alpha,{\mathcal{A}_n}\right)$ represents the set of spectrum allocation rewards achieved by spectrum allocation schemes for each BS $n$ at time $t$. Therefore, the reward function of action $i$ on BS $n$ is given by:
\begin{equation}
\begin{split}
r_{t,n}^{\alpha,i} \left(\boldsymbol{x}_{t,n}^\alpha, {\boldsymbol{a}_{n,i}}\right) =u_{n} \left(\boldsymbol{a}_{n,i}, {\boldsymbol{a}_{-n}}\right),
\end{split}
 \end{equation}
where {$\boldsymbol{x}_{t,n}^\alpha={\boldsymbol{a}_{-n}}$ represents the input at time $t$ of ESN $\alpha$ and ${\boldsymbol{a}_{-n}}$ is the actions that other BSs adopt at time $t$}. 

The reward function in {ESN} $\beta$ allows each BS to choose the optimal spectrum allocation action based on the expected reward given the actions that other BSs adopt. It can be expressed by:
\begin{equation}
\begin{split}
r_{t,n}^{\beta,i}\left( {{\boldsymbol{x}_{t,n}^\beta},{\boldsymbol{a}_{n,i}}} \right) &=\sum\limits_{{\boldsymbol{x}_{t,n}^\alpha} \in {A_{ - n}}} {r_{t,n}^{\alpha,i} \left(\boldsymbol{x}_{t,n}^\alpha, {\boldsymbol{a}_{n,i}}\right) {\pi _{ - n, \boldsymbol{x}_{t,n}^\alpha}}}\\
&\mathop  = \limits^{\left( a \right)} \sum\limits_{{\boldsymbol{a}_{ - n}} \in {\mathcal{A}_{ - n}}} {{r_{t,n}^{\alpha,i} }\left( {{\boldsymbol{a}_{n,i}},{\boldsymbol{a}_{ - n}}} \right){\pi _{ - n, \boldsymbol{a}_{ - n}}}}\\
&={\mathbb{E}}\left[ {{r_{t,n}^{\alpha,i} }\left( {{\boldsymbol{a}_{n,i}}} \right)} \right],\\
\end{split}
\end{equation}
where (a) stems from the fact that 
$\boldsymbol{x}_{t,n}^\alpha$ represents the actions that all BSs, other than BS $n$, take at time $t$, i.e., $\boldsymbol{x}_{t,n}^\alpha=\boldsymbol{a}_{t,-n}$.

\subsection{{Update based on Echo State Networks}}
In this subsection, we first introduce the {update phase, based on the ESN framework,} that each BS $n$ uses to store and estimate the reward of each spectrum allocation scheme. Then, we {present} the proposed ESN-based approach that each BS $n$ uses to choose the optimal spectrum allocation scheme.  
As shown in Fig. \ref{fig.2}, the internal structure of ESNs for BS $n$ consists of three components: a) input weight matrix $\boldsymbol{W}_n^{in}$, b) recurrent matrix $\boldsymbol{W}_n$, and c) output weight matrix $\boldsymbol{W}_n^{out}$. Given these basic definitions, for each BS $n$, an ESN model is essentially a dynamic neural network, known as the dynamic reservoir, which will be combined with the input $\boldsymbol{x}_{t,n}$ representing network state. Therefore, we first explain how an ESN model can be generated. Mathematically, the dynamic reservoir consists of the input weight matrix $\boldsymbol{W}_n^{in} \in {\mathbb{R}^{N \times 2K_n}}$, and the recurrent matrix $\boldsymbol{W}_n \in {\mathbb{R}^{N \times N}}$, where $N$ is the number of units of the dynamic reservoir that each BS $n$ uses to store the users' states. The output weight matrix $\boldsymbol{W}_n^{out} \in {\mathbb{R}^{{\left| {{\mathcal{A}_n}} \right|} \times \left(N+2K_j\right)}}$ includes the linear readout weights and is trained to approximate the reward function of each BS. The BS's reward function  essentially reflects the rate achieved by that BS. The dynamic reservoir of BS $n$ is therefore given by the pair $\left( \boldsymbol{W}_n^{in}, \boldsymbol{W}_n \right)$  and $\boldsymbol{W}_n$ is defined as a sparse matrix with a spectral radius less than one \cite{27}. $\boldsymbol{W}_n^{in}, \boldsymbol{W}_n$, and $\boldsymbol{W}_n^{out}$ are initially generated randomly by uniform distribution. In this ESN model, one needs to only train $\boldsymbol{W}_n^{out}$ to approximate the reward function which illustrates that ESNs are easy to train \cite{26,27,28}. Even though the dynamic reservoir is initially generated randomly, it will be combined with {the} input to store the users' states and it will also be combined with the trained output matrix to approximate the reward function.

Since the users' associations change depending on the spectrum allocation scheme that each BS adopts, the ESN model of each BS $n$ needs to update its input $\boldsymbol{x}_{t,n}$ and store the users' states, which is done by the dynamic reservoir state ${\boldsymbol{\mu  }_{t,n}}$. Here, ${\boldsymbol{\mu}_{t,n}}$ denotes the users' association results and the users' states for each BS $n$ at time $t$. The dynamic reservoir state for each BS $n$ can be computed as follows:
\begin{equation}\label{eq:utn}
{\boldsymbol{\mu}_{t,n}^j} ={\mathop{f}\nolimits}\!\left( {\boldsymbol{W}_n^j{\boldsymbol{\mu}_{t - 1,n}^j} + \boldsymbol{W}_n^{j,in}{\boldsymbol{x}_{t,n}^j}} \right),
\end{equation}
where $f\left(  \cdot  \right)$ is the tanh function and $j \in \left\{ {\alpha ,\beta } \right\}$. Suppose that, each BS $n$, has ${{\left| {{\mathcal{A}_n}} \right|}}$ spectrum allocation actions, $\boldsymbol{a}_{n,1}, \ldots, {\boldsymbol{a}_{n,{{\left| {{\mathcal{A}_n}} \right|}} }}$, to choose from. Then, the {ESN scheme} will have ${{\left| {{\mathcal{A}_n}} \right|}} $ outputs, one corresponding to each one of those actions. We must train the output matrix $\boldsymbol{W}_n^{j,out}$, so that the output $i$ yields the value of the reward function $r_{t,n}^{j,i} \left(\boldsymbol{x}_{t,n}^j, {\boldsymbol{a}_{n,i}}\right)$ due to action ${\boldsymbol{a}_{n,i}}$ in the input $\boldsymbol{x}_{t,n}^j$:
\begin{equation}\label{eq:rtn}
r_{t,n}^{j,i} \left(\boldsymbol{x}_{t,n}^j, {\boldsymbol{a}_{n,i}}\right) = {\boldsymbol{W}_{t,in}^{j,out}}\left[ {{\boldsymbol{\mu}_{t,n}^j};{\boldsymbol{x}_{t,n}^j}} \right],
\end{equation}
where $\boldsymbol{W}_{t,in}^{j,out}$ denotes the $i$th row of $\boldsymbol{W}_{t,n}^{j,out}$. (\ref{eq:rtn}) is used to estimate the reward of each BS $n$ that adopts any spectrum allocation action after training $\boldsymbol{W}_n^{j,out}$. To train $\boldsymbol{W}_n^{j,out}$, a linear gradient descent approach can be used to derive the following update rule:
\begin{equation}\label{eq:w}
{\boldsymbol{W}_{t + 1,in}^{j,out}}\! = \!{\boldsymbol{W}_{t,in}^{j,out}} + {\lambda_j}\! \left( {e_{t,n}^{j,i} \!-\!r_{t,n}^{j,i} \!\left(\boldsymbol{x}_{t,n}^j, {\boldsymbol{a}_{n,i}} \right)} \!\right)\!\!\left[ {{\boldsymbol{\mu}_{t,n}^j};{\boldsymbol{x}_{t,n}^j}} \right]^{\mathrm{T}},
\end{equation}
where ${\lambda_j}$ is the learning rate for {ESN} $j$ and $e_{t,n}^{j,i}$ is the $j$th actual reward at action $i$ of BS $n$ at time $t$, i.e., $e_{t,n}^{\alpha,i}=u_n^i\left( {{\boldsymbol{a}_{n,i}},{\boldsymbol{a}_{ - n}}} \right)$ and $e_{t,n}^{\beta,i}={\mathbb{E}}\left[ {{u_n}\left( {{\boldsymbol{a}_n}} \right)} \right] $. Note that ${\boldsymbol{a}_{ - n}}$ denotes the actions that all BSs other than BS $n$ adopt now.

\subsection{Reinforcement learning with ESNs algorithm}

To solve the game, we introduce {an} ESN-based reinforcement learning approach to find the mixed strategy NE. The proposed ESN-based reinforcement learning approach consists of {two phases:} ESN $\alpha$ and ESN $\beta$. ESN $\alpha$ is used to approximate the utility function of our proposed game. ESN $\alpha$ stores the reward of utility function at any case which can be used by ESN $\beta$. ESN $\beta$ uses the reward stored in ESN $\alpha$ to find the mixed strategy NE {using} a reinforcement learning approach. In our proposed algorithm, each SBS $n$ needs to calculate the {fraction} of time slots $L$ {on the unlicensed band that is allocated to the LTE-U network based on (8)}, update the users' associations and store users' states based on (\ref{eq:utn}), estimate the rewards of spectrum allocation actions based on (\ref{eq:rtn}), choose the optimal allocation scheme, and update the output matrix $\boldsymbol{W}_n^{out}$ based on (\ref{eq:w}) at each time. 

In order to guarantee that any action always has a {non-zero} probability {to be chosen}, the $\varepsilon$-greedy exploration \cite{30} is adopted in the proposed algorithm. {This} mechanism is responsible for selecting the actions that {each} agent will perform during the learning process {while harmonizing} the tradeoff between exploitation and exploration. Therefore, the probability of BS $n$ playing action $i$ {will be} given by: 
\begin{equation}\small
\Pr \!\left( {{\boldsymbol{a}_n}\left( t \right) \!=\! {\boldsymbol{a}_{n,i}}} \right) \!=\! \left\{ {\begin{array}{*{20}{c}}
{\!\!\!\!\!1 - \varepsilon  + \frac{\varepsilon }{{\left| {{\mathcal{A}_n}} \right|}},\;{\boldsymbol{a}_{n,i}} \!=\! \arg\! \mathop {\max }\limits_{{\boldsymbol{a}_n} \in {\mathcal{A}_n}}\! {{r}_{t,n}^{\beta}}\!\left( \!{{\boldsymbol{x}_{t,n}^\beta},{\boldsymbol{a}_n}} \!\right)\!\!,}\\
{\!\!\!\!\!\!\!\!\!\!\!\!\!\!\!\!\!\!\!\!\!\!\frac{\varepsilon }{{\left| {{\mathcal{A}_n}} \right|}},\;\;\text{otherwise}.\;\,\,\,\;\;\;\;\;\;\;\;\;}
\end{array}} \right. 
\end{equation}
The $\varepsilon$-greedy mechanism decides the probability distribution {over the} action set for each BS. {Naturally,} the probability distribution {over} each SBS's action set consists of {two components: a large probability corresponding to the optimal action and a small equal probability for other actions.} Thus, based on {the} $\varepsilon$-greedy mechanism, each BS can obtain the probability distribution {over the action sets of other BSs} by {observing only} the spectrum allocation action that results in the optimal reward.

The learning rates in ESNs have two different rules: a) fixed value and b) the Robbins-Monro conditions \cite{26}. The Robbins-Monro conditions can be given by: 
\begin{equation}
\left\{ {\begin{array}{*{20}{c}}
{\left( i \right)\;\;\;\;\;\;\;\;{\lambda _j }\left( n \right) > 0,\;\;\;\;\;\;\;\;\;\;\;\;\;\;\;\;}\\
{\left( {ii} \right)\;\;\;\;\;\;\;\mathop {\lim }\limits_{t \to \infty } \sum\limits_{n = 1}^t {{\lambda _j }\left( n \right)}  =  + \infty },\\
{\left( {iii} \right)\;\;\;\;\;\;\;\mathop {\lim }\limits_{t \to \infty } \sum\limits_{n = 1}^t {{\lambda _j^2 }\left( n \right)}  <  + \infty },
\end{array}} \right.
\end{equation}
where $j \in \left\{ {\alpha ,\beta } \right\}$. The learning rate has an effect on the speed of the convergence of our proposed algorithm. However, the two learning rules will converge, as will be shown later in this section.  

We assume that each user can only connect to one BS in the uplink or downlink at each time. We further {consider that each SBS knows the fraction of the time slots on the unlicensed band that is allocated to the LTE-U network}. Based on the above formulations, the distributed RL approach based on ESNs performed by every BS $n$ is shown in Algorithm 1. In line 8 of Algorithm 1, we capture the fact that each SBS broadcasts the action that it adopts now and the probability distribution of action profiles to other BSs.

\begin{algorithm}[!t]\footnotesize
\caption{Reinforcement learning with ESNs}   
\label{alg:Framwork}   
\begin{algorithmic} [1] 
\REQUIRE The set of users' association states, $\boldsymbol{x}_{t,n}^\alpha$ and $\boldsymbol{x}_{t,n}^\beta$;\\ 
\ENSURE initialize $\boldsymbol{W}_n^{\alpha,in}$, $\boldsymbol{W}_n^\alpha$, $\boldsymbol{W}_n^{\alpha,out}$, $\boldsymbol{W}_n^{\beta,in}$, $\boldsymbol{W}_n^\beta$, $\boldsymbol{W}_n^{\beta,out}$, $\boldsymbol{r}_{0,n}^\alpha=0$, $\boldsymbol{r}_{0,n}^\beta=0$  \\ 
\STATE calculate fraction of the time slots $L$ based on (8)
\FOR {time $t$} 
\IF{$rand(.) < \varepsilon$}
\STATE randomly choose one action 
\ELSE
\STATE choose action $\boldsymbol{a}_{n,i}\left(t\right) = \mathop {\arg \max }\limits_{\boldsymbol{a}_{n,i}\left(t\right) } \left( {r_{t,j}^{\beta,i}\left( {{\boldsymbol{x}_{t,n}^{\beta}},\boldsymbol{a}_{n,i}\left(t\right)} \right)} \right)$ 
\ENDIF 
\STATE broadcast the action $\boldsymbol{a}_{n,i}\left(t\right)$ and the optimal action $\boldsymbol{a}_{n,*}$ that results in the maximal value of reward function
\STATE calculate the reward $r^\alpha$ and $r^\beta$ based on (\ref{eq:utn}) and (\ref{eq:rtn}) 
\STATE update the output weight matrix $\boldsymbol{W}_{t,ij}^{\alpha,out}$ and $\boldsymbol{W}_{t,ij}^{\beta,out}$ based on (\ref{eq:w})
\ENDFOR  
\end{algorithmic}
\end{algorithm}  

In essence, at every time instant, {each} BS allocates its spectrum to the users and maximizes its own rate. The users {will then} send {a} connection request to all BSs at each time. {After the proposed algorithm converges to the mixed strategy NE}, each user could get the best rate and each BS maximizes its total rate. Note that the performance of the proposed algorithm can be improved by incorporating a training sequence to update the output weight matrix $\boldsymbol{W}^{out}$. Adjusting the input weight matrix $\boldsymbol{W}^{in}$ and {the} recurrent matrix $\boldsymbol{W}$ appropriately will also improve the accuracy of the algorithm. Algorithm 1 continues to iterate until each user achieves maximal rate and the users' association states remain unchanged.
{The interaction between BSs is independent of the network size and incurs no notable overhead because, in each iteration, each BS needs to only broadcast the optimal action and the action it has currently selected. The complexity of the proposed algorithm is $O(\left| {{\mathcal{A}_{1}}}\right| \times \left| {{\mathcal{A}_{2}}}\right| \times \dots \times \left| {{\mathcal{A}_{N_b}}}\right| )$. This is due to the fact that the worst case for each BS is to traverse all of the possible actions in its action space. However, the proposed ESN-based algorithm is a learning algorithm which can record the utility value of the resource allocation schemes that the BSs have been used, which will greatly reduce the number of iterations. Moreover, after training, the proposed ESN-based algorithm can automatically choose the optimal resource allocation scheme without traversing all of the possible actions again. Therefore, the proposed algorithm can be implemented with reasonable complexity, as will be further shown in the simulations in Section \ref{Section:simulation}.} 

\subsection{Convergence of the ESN-based algorithm}
In the proposed algorithm, we use {ESN} $\alpha$ to store and estimate the reward of each spectrum allocation scheme. Then, we train {ESN} $\beta$ as RL algorithm to solve the proposed game. Since ESNs have two different learning rules to update the output matrix, in this subsection, we prove the convergence of the two learning phases of our proposed algorithm. We first prove the convergence of ESNs with the Robbins-Monro learning rule. Next, we use the continuous time version to prove the convergence of ESNs with the fixed value learning rule. Finally, we prove the proposed algorithm reaches to the mixed strategy NE.

\begin{theorem}\emph{ ESN $\alpha$ and ESN $\beta$ for each BS $n$ converge to the utility function $u_n$ and ${\mathbb{E}}\left[ {{u_n}\left( {{\boldsymbol{a}_n}} \right)} \right]$ with probability 1 when $\lambda_j \ne {1 \mathord{\left/
 {\vphantom {1 t}} \right.
 \kern-\nulldelimiterspace} t}$.}
\end{theorem}

\begin{proof} In order to prove this theorem, we first need to prove {that} the ESNs {converge}. Then, we formulate the exact value to which the ESNs converge. 

Based on the Gordon's Theorem \cite{26}, the ESNs converge with probability 1 must satisfy: a) a finite {Markov decision problems (MDPs)}, b) Sarsa learning \cite{31} is being used with a linear function approximator, and c) learning rates satisfy the Robbins-Monro conditions (${\lambda} > 0,\sum\nolimits_{t = 0}^\infty  {{\lambda\left(t\right)} = +\infty ,\sum\nolimits_{t = 0}^\infty  {\lambda^2\left(t\right) <+ \infty } } $). The proposed game only has two states that are request state and stable state, and the number of actions is finite which satisfies a).

From (\ref{eq:rtn}) and (\ref{eq:w}), we can formulate the update equation for ESNs as follows:   
\begin{equation}\label{eq:rc}
r_{t + 1,n}^{j,i} - r_{t,n}^{j,i} = {\lambda _j}\left(u_{n}^{j,i} - r_{t,n}^{j,i} \right), \hspace{1em} j \in \left\{ \alpha, \beta \right\}.
\end{equation}
Actually, (\ref{eq:rc}) is a special form of Sarsa(0) learning \cite{31}. Condition c) is trivially satisfied via our learning scheme's definition. Therefore, we can conclude that the ESNs in our game with the Robbins-Monro learning rule satisfy Gordon's Theorem and converge with probability 1. However, Gordon's Theorem dose not formulate the exact value to which the ESNs converge. Therefore, we use the continuous time version of (\ref{eq:rc}) to formulate the exact value to which the ESNs converges.

To obtain the continuous time version, consider $\Delta t \in [0,1]$ to be a small amount of time and $r_{t +\Delta t ,n}^i - r_{t,n}^i  \approx \Delta t  \times {\alpha _t}\left( {u_{t+\Delta t,n}^i - r_{t,n}^i} \right)$ to be the approximate growth in $r_{n}^i$ during $\Delta t$. Dividing both sides of the equation by $\Delta t$ and taking the limit for $\Delta t \to 0$, (\ref{eq:rc}) can be expressed as follows:
\begin{equation}\label{eq:dr}
\frac{{dr_{t,n}^{j,i}}}{{dt}} \approx {\lambda _j}\left( {u_n^{j,i} - r_{t,n}^{j,i}} \right).
\end{equation}
The general solution for (\ref{eq:dr}) can be found by integration:
\begin{equation}\label{eq:rC}
r_n^{j,i} = C\exp \left( { - {\lambda _j}t} \right) + u_n^{j,i}, \;\; j \in \left\{ \alpha, \beta \right\},
\end{equation}
where $C$ is the constant of integration. As $\exp\left(-x\right)$ is {a} monotonic function and ${\lim _{t \to \infty }}\!\exp( {\! - \alpha_tt}) \\= 0$, when $\lambda_j \ne {1 \mathord{\left/{\vphantom {1 t}} \right.\kern-\nulldelimiterspace} t}$. It is easy to observe that, when $t \to \infty $, the limit of (\ref{eq:rC}) is given by:
\begin{equation}\label{eq:limr}
{\lim _{t \to \infty }}r_n^{j,i}=\left\{ {\begin{array}{*{20}{c}}
{C\exp \left( { - 1} \right) + u_n^{j,i},\;\;\;\;{\lambda _j} = \frac{1}{t}},\\
\;\;\;\;\;\;\;\;\;\;\;\;{u_n^{j,i},\;\;\;\;\;\;\;\;\;\;\;\;\;{\lambda _j} \ne \frac{1}{t}}.
\end{array}} \right.
\end{equation}
From (\ref{eq:limr}), we can conclude that the ESNs converge to the utility function as ${\lambda _j} \ne \frac{1}{t}$, however, as ${\lambda _j} =\frac{1}{t}$, the ESNs converge to the utility function with a constant $C\exp \left( { - 1} \right)$. Moreover, we can see that the convergence of ESNs actually has no relationship with learning rate, it only needs enough time to update. 
This completes the proof.
\end{proof} 

\begin{theorem}\emph{ The ESN-based algorithm converges to a mixed Nash equilibrium, with the mixed strategy probability ${\boldsymbol{\pi}_n^*} \in \Delta\left(\mathcal{A}_n\right)$, $\forall n \in \mathcal{B}$.}
\end{theorem}
\begin{proof} In order to prove Theorem 2, we need to establish the mixed NE conditions in (\ref{eq:mNE}). We assume that the spectrum allocation action {$\boldsymbol{a}_{n,*}$} results in the optimal reward given the optimal mixed strategy $(\boldsymbol{\pi}_{n}^*,\boldsymbol{\pi}_{-n}^*)$, which means that $\pi_{n,\boldsymbol{a}_{n,*}}^*=1 - \varepsilon  + \frac{\varepsilon }{\left| \mathcal{A}_n \right|}$ and $\pi_{n,\boldsymbol{a}_{n,'}}^*= \frac{\varepsilon }{{{\left| \mathcal{A}_n \right|}}}$. We also assume that {$\boldsymbol{a}_{n} \in \mathcal{A}_n/\boldsymbol{a}_{n,*}$} results in the optimal reward given the optimal mixed strategy $(\boldsymbol{\pi}_{n}, \boldsymbol{\pi}_{-n}^*)$. Based on the Theorem 1, ESN $\beta$ of the proposed algorithm converges to ${\mathbb{E}}\left[ {{u_n}\left( {{\boldsymbol{a}_n}} \right)} \right]$. Thus, (\ref{eq:mNE}) can be rewritten as follows:
\begin{equation}\label{eq:proof}
\begin{split}
&{\tilde u_n}\left( {\boldsymbol{\pi} _n^*,\boldsymbol{\pi} _{ - n}^*} \right)-{\tilde u_n}\left( {{\boldsymbol{\pi} _n},\boldsymbol{\pi} _{ - n}^*} \right)\\
&=\sum\limits_{{\boldsymbol{a}_{ n}} \in {\mathcal{A}_{ n}}}[\pi_{n,\boldsymbol{a}_n}^* \sum\limits_{{\boldsymbol{a}_{ - n}} \in {\mathcal{A}_{ - n}}}{{u_n}\left( {{\boldsymbol{a}_n},{\boldsymbol{a}_{ - n}}} \right)\pi_{n,{\boldsymbol{a}_{- n}}}^*}-{\pi_{n,\boldsymbol{a}_n} \sum\limits_{{\boldsymbol{a}_{ - n}} \in {\mathcal{A}_{ - n}}}{{u_n}\left( {{\boldsymbol{a}_n} ,{\boldsymbol{a}_{ - n}}} \right)\pi_{n,{\boldsymbol{a}_{- n}}}^*} }]\\
&=\sum\limits_{{\boldsymbol{a}_{ n}} \in {\mathcal{A}_{ n}}}\left[{\pi_{n,\boldsymbol{a}_n}^*{\mathbb{E}}\left[ {{u_n}\left( {{\boldsymbol{a}_n}} \right)} \right]-\pi_{n,\boldsymbol{a}_n}{\mathbb{E}}\left[ {{u_n}\left( {{\boldsymbol{a}_n}} \right)} \right]} \right]\\
&\mathop = \limits^{\left( a \right)}\left(1-\epsilon\right)\left({{\mathbb{E}}\left[ {{u_n}\left( {{{\boldsymbol{a}_{n,*}}}} \right)} \right]-{\mathbb{E}}\left[ {{u_n}\left( {{\boldsymbol{a}_n}} \right)} \right]} \right)\\
\end{split}
\end{equation}
where (a) is obtained from the fact that $\pi_{n,{\boldsymbol{a}_{n,*}}}^*=\pi_{n,\boldsymbol{a}_n} =1 - \varepsilon  + \frac{\varepsilon }{\left| \mathcal{A}_n \right|}$ and $\pi_{n,{\boldsymbol{a}_{n,'}}}^*=\pi_{n,{\boldsymbol{a}_{n,''}}}= \frac{\varepsilon }{{{\left| \mathcal{A}_n \right|}}}$, ${{{\boldsymbol{a}_{ n,'}}, {\boldsymbol{a}_{n,''}} }\in{{{\mathcal{A}_n}} \mathord{\left/
 {\vphantom {{{A_n}} a}} \right.
 \kern-\nulldelimiterspace} \boldsymbol{a}_{n,*}}}$. 
 Since in the proposed algorithm, the optimal action {${{\boldsymbol{a}_{n,*}}}$} results in the optimal {${\mathbb{E}}\left[ {{u_n}\left( {{\boldsymbol{a}_{n,*}}} \right)} \right]$}, we can conclude that ${{\mathbb{E}}\left[ {{u_n}\left( {{{\boldsymbol{a}_{n,*}}}} \right)} \right]-{\mathbb{E}}\left[ {{u_n}\left( {{\boldsymbol{a}_n}} \right)} \right]}\ge 0$. This completes the proof. 
 \end{proof}

\section{Simulation results}\label{Section:simulation} 
In this section, we evaluate the performance of the proposed ESN algorithm using simulations. We first introduce the simulation parameters. {Then}, we evaluate the performance of {the ESN-based} approximation and estimation {phases} in our proposed algorithm. Finally, we show the improvements in terms of {the} sum-rate of all users and the $50$th percentile of users by comparing the proposed algorithm with three {baseline algorithms based on Q-learning.}
\subsection{System {parameters}}
For our simulations, we use the following parameters. One macrocell with a radius $r_M=500$ meters is considered with $N_s$ uniformly distributed picocells, $W=2$ uniformly distributed {WAPs}, and $U$ uniformly distributed LTE-U users. The {SBSs} share the licensed band with {the} MBS and the unlicensed band with WiFi. Each {WAP} has 4 WiFi users. The channel gain {follows a Rayleigh distribution} with unit variance. The WiFi network is set up based on the IEEE 802.11n protocol {operating} at the $5$ GHz band with a RTS/CTS mechanism. {The path loss models for LTE and LTE-U are based on \cite{6}.} Other parameters are listed in Table  \uppercase\expandafter{\romannumeral1}. The results are compared to three schemes: {a) Q-learning with uplink-downlink decoupling applied to an LTE-U system, b) Q-learning with uplink-downlink decoupling within an LTE system, and c) Q-learning without uplink-downlink decoupling applied to an LTE-U system.} All statistical results are averaged over {5000} independent runs. Hereinafter, we use the term ``sum-rate'' to refer to the total downlink and uplink rates. 

{In our simulations, we} assume {that each Q-learning approach} has knowledge of the action that each BS takes and the entire users' information including interference and location. The update rule of {each Q-learning approach} will be given by:
\begin{equation}
{Q_{t + 1}^i} = \left( {1 - \lambda_q } \right){Q_t^i} +{\lambda_q }\left( {{\boldsymbol{a}_{n,i}}, {{\boldsymbol{a}_{ - n,*}}}} \right),
\end{equation}
where {${\boldsymbol{a}_{ - n,*}}$} represents the optimal action profile of all BSs other than SBS $n$.

\begin{table}\footnotesize
 \caption{
    \vspace*{-0em}SYSTEM PARAMETERS \cite{6}}\vspace*{-1em}
\centering  
\begin{tabular}{|c|c|c|c|}
\hline
\textbf{Parameters} & \textbf{Values} & \textbf{Parameters} & \textbf{Values} \\
\hline
$P_M $ & 43 dBm & $P_P$ & {30} dBm\\
\hline
$P_u$ & 20 dBm & $\alpha$ & 0.05\\
\hline
$\varepsilon$ & 0.7 & $N$ & 1000\\
\hline
{$T_\sigma$} & {9 $\mu$s} & $\delta $ & 0 $\mu$s\\
 \hline
 $F_l^\textrm{UL}$ & 10 MHz & $F_l^\textrm{DL}$ & 10 MHz \\
\hline
$F_u$ & 20 MHz & $\mathcal{L}_b$ & 100 m \\
\hline
$E\left[P\right] $ & 1500 byte & SIFS & 16 $\mu$s\\
\hline
 CTS & 304 {bits}&DIFS & {34} $\mu$s\\
\hline
ACK & 304 {bits} & RTS & 352 {bits}\\ 
\hline
$C$ & 130 Mbps& $\lambda_q, \lambda_\beta$ & 0.06, 0.06 \\
\hline
$\lambda_\alpha$ & 0.08 & $Z$ & 10 \\
\hline
Licensed path loss &\!\!$15.3\!+\!\!37.5\!\log_{10}(m)$\!\! & {$\eta$} & {0.7} \\
\hline
Unlicensed path loss &\!\!$15.3\!+\!\!50\!\log_{10}(m)$\!\!& {$H$} & {416 bits}\\
\hline
\end{tabular}
\vspace{-0.5cm}
\end{table}

\subsection{ESNs {approximation}}

\begin{figure}[!t]
  \begin{center}
   \vspace{0cm}
    \includegraphics[width=9cm]{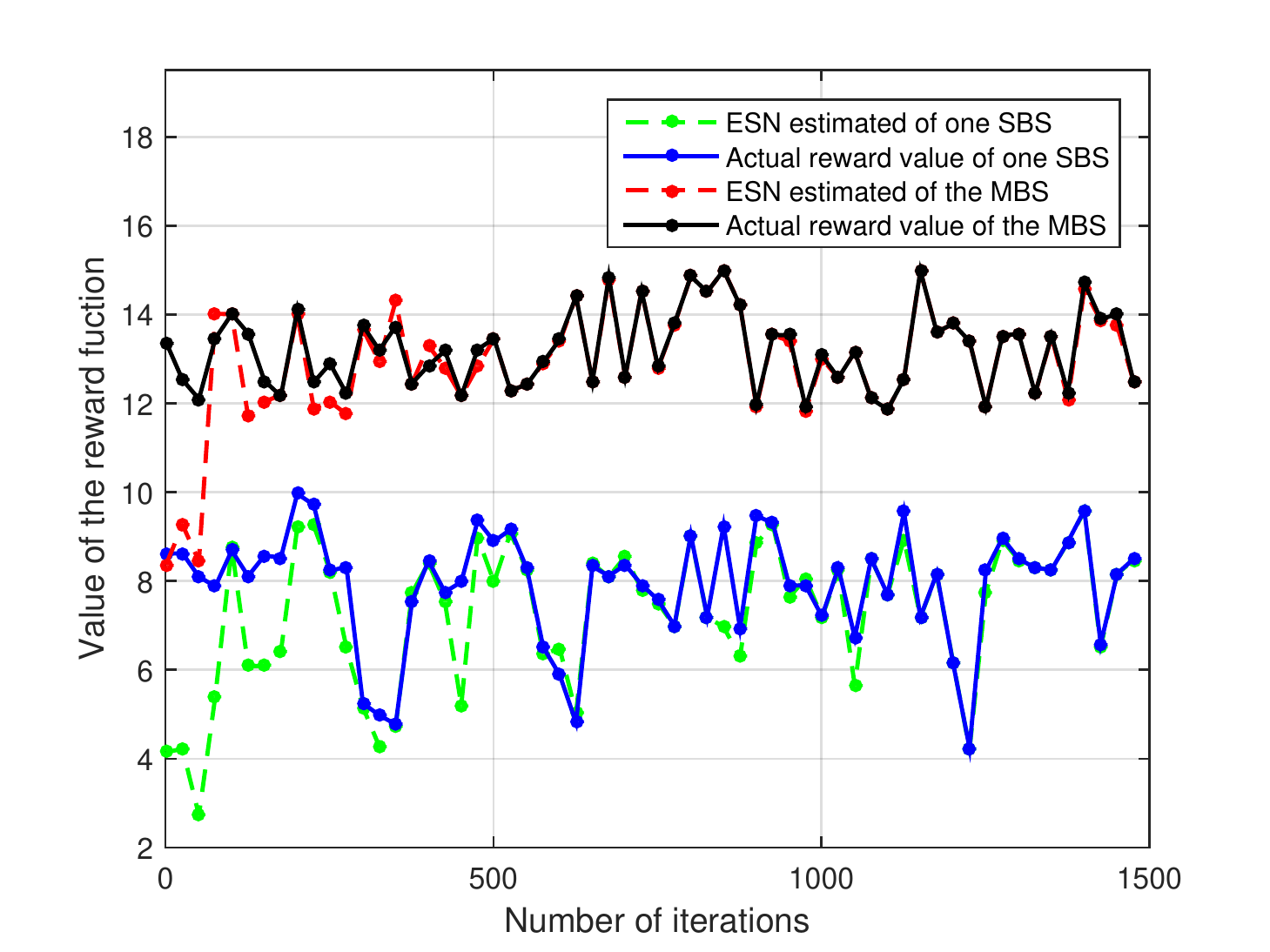}
    \vspace{-0.3cm}
    \caption{\label{fig2} {Approximation of the reward function by ESN $\alpha$} ($N_b=5, U=40, R_w=4$ Mbps).}
  \end{center}\vspace{-0.6cm}
\end{figure}

Fig. \ref{fig2} shows how the {ESN $\alpha$ approximates} the reward function as the number of iteration varies. In Fig. \ref{fig2}, we can see that, as the number of iteration increases, the approximations of {all} BSs improve. This demonstrates the effectiveness of ESN as an approximation tool for utility functions. Moreover, the network state stored in the {ESN} improves the approximation which {updates} the value of {the} reward function according to the network state. Fig. \ref{fig2} also shows that {the proposed approach requires only $500$} iterations to {approximate} the reward function. This is due to the fact that ESN {needs to only} train the output matrix which reduces the {training process}.  

\begin{figure*}[htbp]
  \centering
  {\subfigure[$\lambda_\alpha=0.08$ ]{\includegraphics[width=5.5cm]{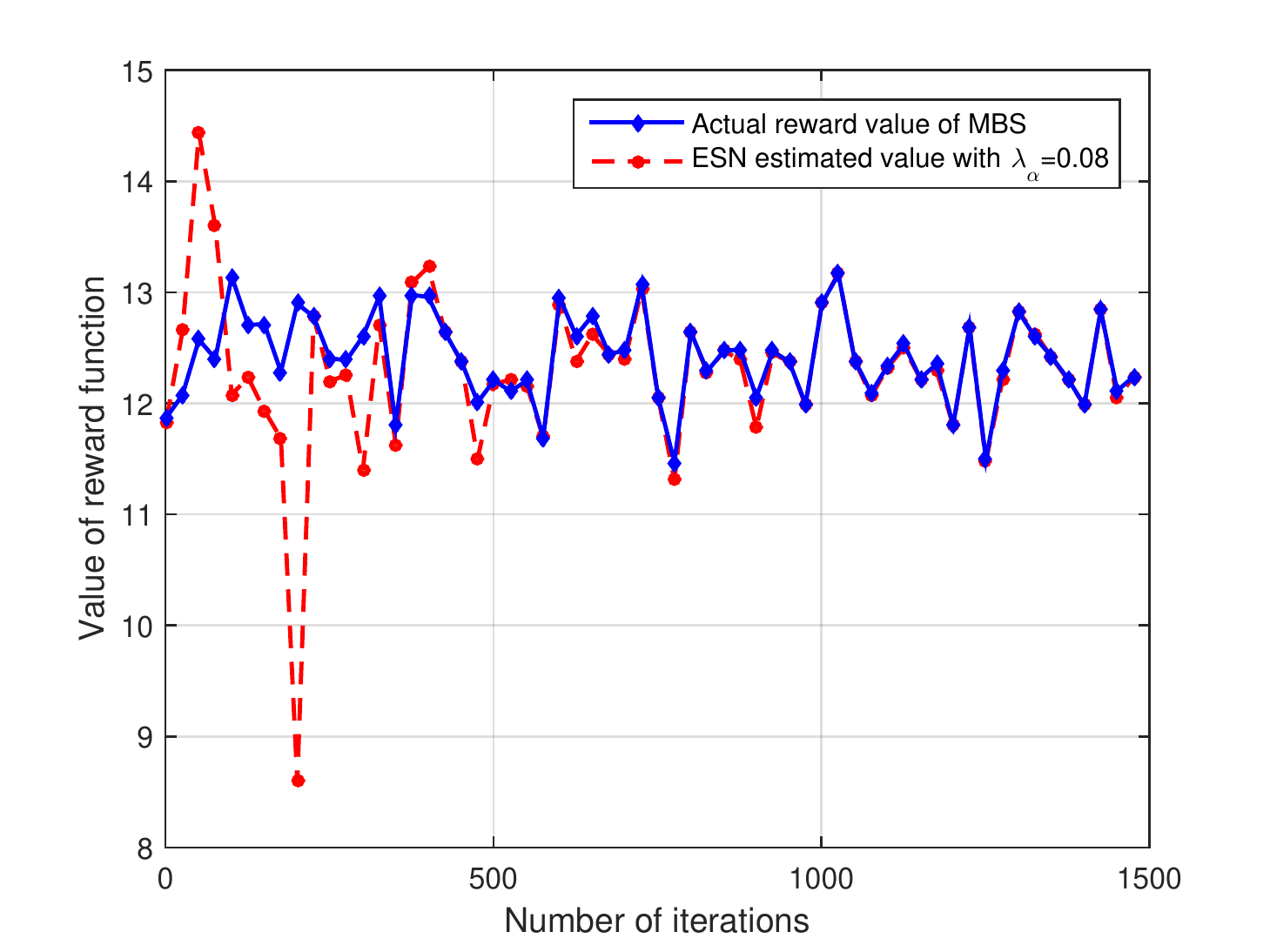}
\label{fig3a}}\hspace{-0.65cm}
\subfigure[$\lambda_\alpha=0.04$ and $\lambda_\alpha=0.01$]{\includegraphics[width=5.5cm]{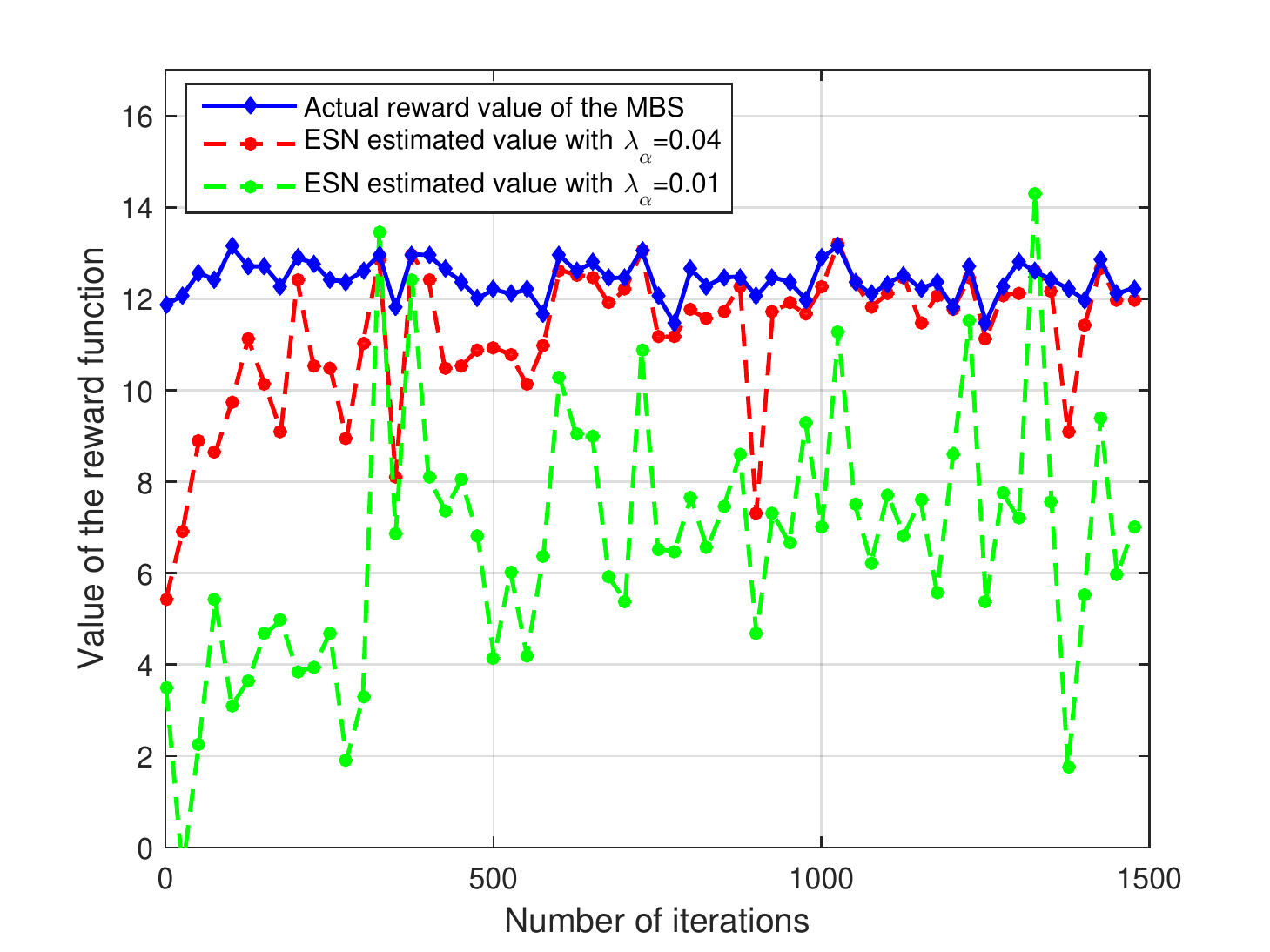}
\label{fig3b}}\hspace{-0.65cm}
\subfigure[$\lambda_\alpha=0.15$]{\includegraphics[width=5.5cm]{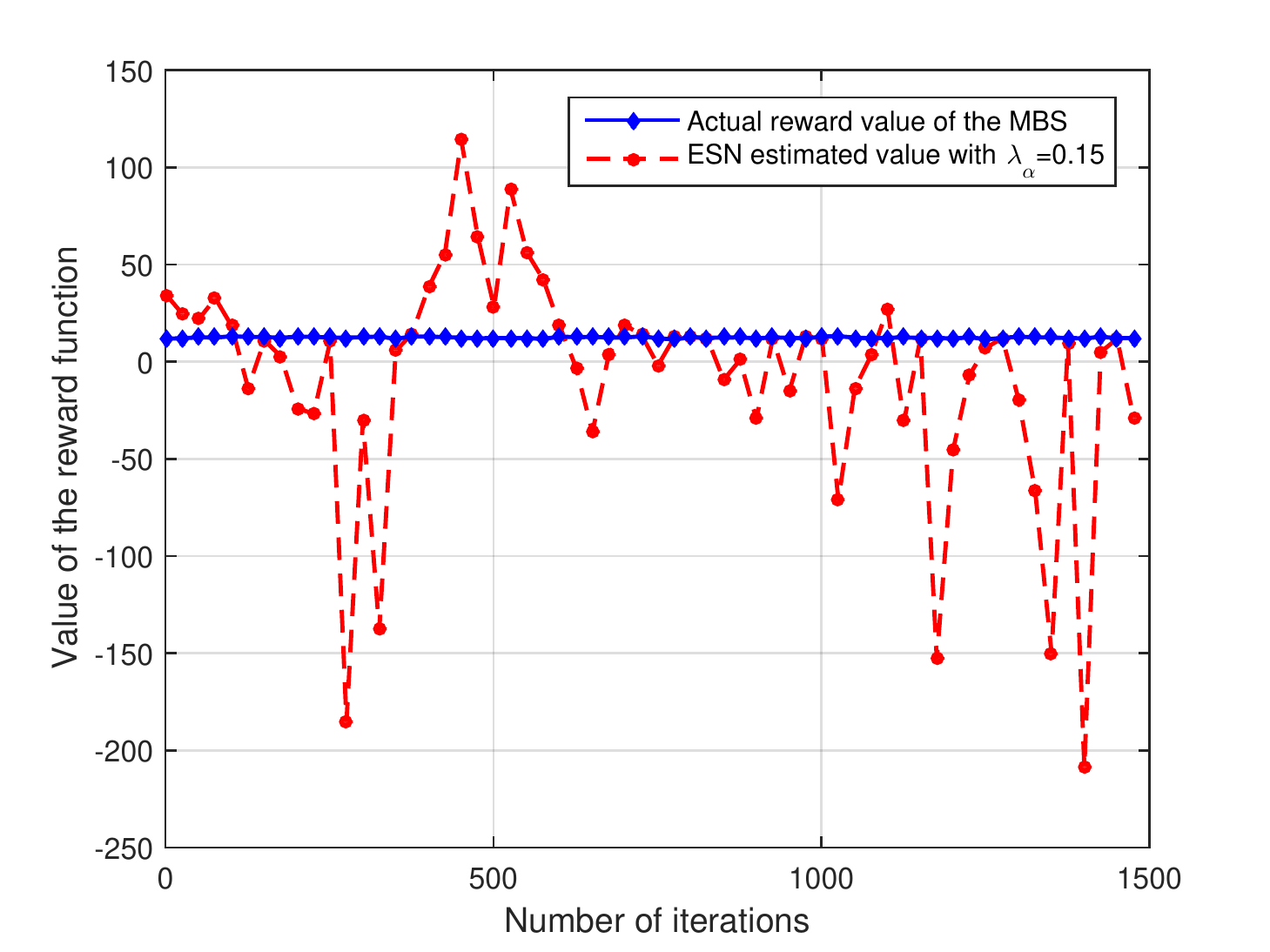}
\label{fig3c}}}
  \caption{An SCN using ESNs with uplink-downlink decoupling in LTE-U. The left hand shows the proposed ESNs learning algorithm running in SBS. The right hand shows the architecture of {an ESN}.}\label{fig3}
\vspace{-0.6cm}
\end{figure*}


In Fig. \ref{fig3}, we show how {ESN} $\alpha$ can approximate the reward function as the learning rate $\lambda_\alpha$ varies. Fig. \ref{fig3a} and Fig. \ref{fig3b} show that, even if the learning rate $\lambda_\alpha$ {changes slightly} from $0.04$ to $0.08$, the {proposed ESN approach} achieves more than $100\%$ improvement in terms of the approximation speed. This is due to the fact that the learning rate {directly affects} the step length of {the ESN adjustment}. However, by comparing Fig. \ref{fig3a} with Fig. \ref{fig3c}, we can see that the approximation of {ESN} $\alpha$ with $\lambda_\alpha=0.15$ {requires} more than {1500} iterations to approximate the reward function, while, for $\lambda_\alpha=0.08$, it only needs {500} iterations. Clearly, when the learning rate $\lambda_\alpha$ is too large, the update value for the output matrix of {ESN} $\alpha$ is also large, which results in a low speed of convergence. Therefore, we can conclude that {the choice of an appropriate learning rate is an important factor that affects} the convergence speed of {ESN} $\alpha$. 

\begin{figure}[!t]
  \begin{center}
   \vspace{0cm}
    \includegraphics[width=9cm]{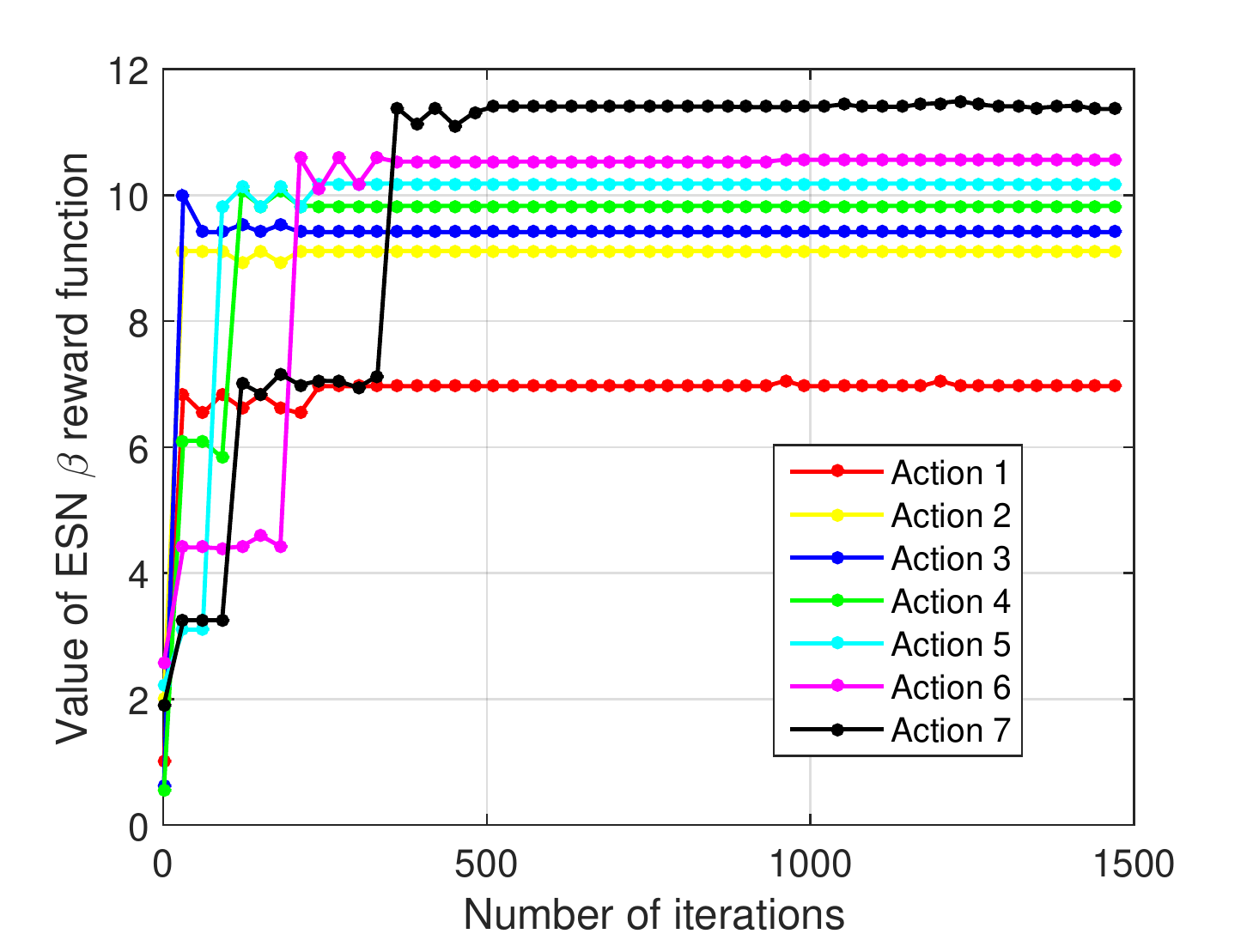}
    \vspace{-0.3cm}
    \caption{\label{fig4} {Performance of the ESN $\beta$ phase of the proposed approach. Each curve corresponds to a different action chosen by the SBS} ($N_b=5, U=40, R_w=4$ Mbps).}
  \end{center}\vspace{-1.2cm}
\end{figure}

In Fig. \ref{fig4}, we show how {the ESN $\beta$ phase of the proposed algorithm allows updating} the expected reward for {one SBS} when it adopts different spectrum allocation actions as the number of iteration varies. {We choose the SBS from the SBSs. Each curve in this figure corresponds to} one spectrum allocation action of {the SBS}. We can see that each {curve} of {ESN} $\beta$ converges to a stable value as the number of iterations increases which implies that by using {ESN}, the SBS can estimate the reward before {deciding on} any action. This is due to the fact that, as {the number of iterations} increases, the approximation of {ESN} $\alpha$ provides the reward that {ESN} $\beta$ {will use} to calculate the expected reward. Fig. \ref{fig4} also shows that, below iteration {$300$}, the expected reward of each spectrum allocation action changes quickly. However, each curve {oscillates only slightly} as the number of iterations is more than {$300$}. The main reason behind this is that, at the beginning, ESN $\alpha$ {requires} some iterations to approximate the reward function. Since {ESN} $\alpha$ has not yet approximated the reward function well, {ESN} $\beta$ can not calculate the expected reward for each action accurately. As the number of iterations {goes} above {500}, {ESN} $\alpha$ {completes} the approximation of {the} reward function, which results in the accurate calculation of {the} expected reward for {ESN} $\beta$.   

\begin{figure}
\centering
\vspace{0cm}
\subfigure[]{
\label{fig7a} 
\includegraphics[width=9cm]{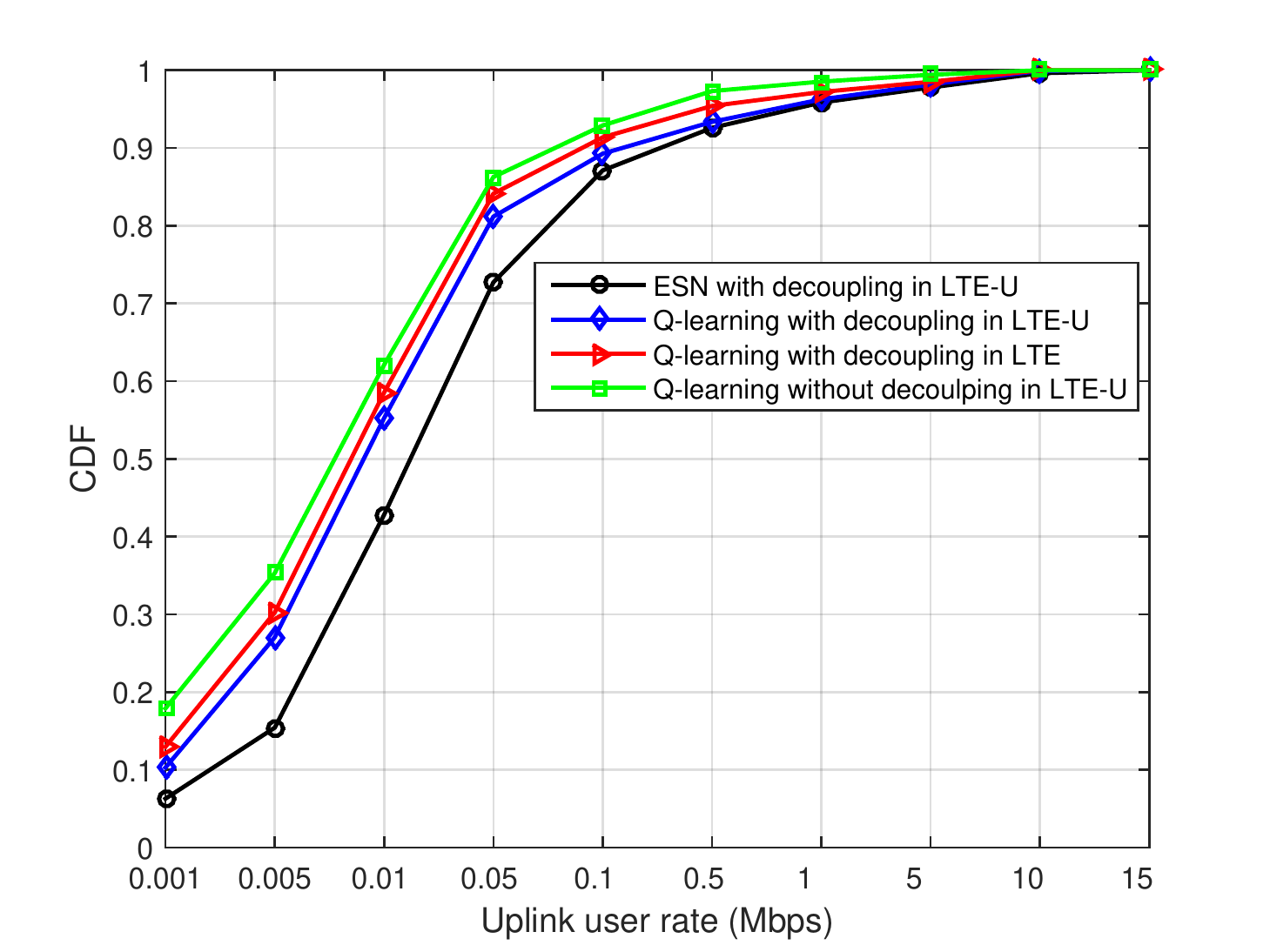}}
\subfigure[]{ 
\label{fig7b} 
\includegraphics[width=9cm]{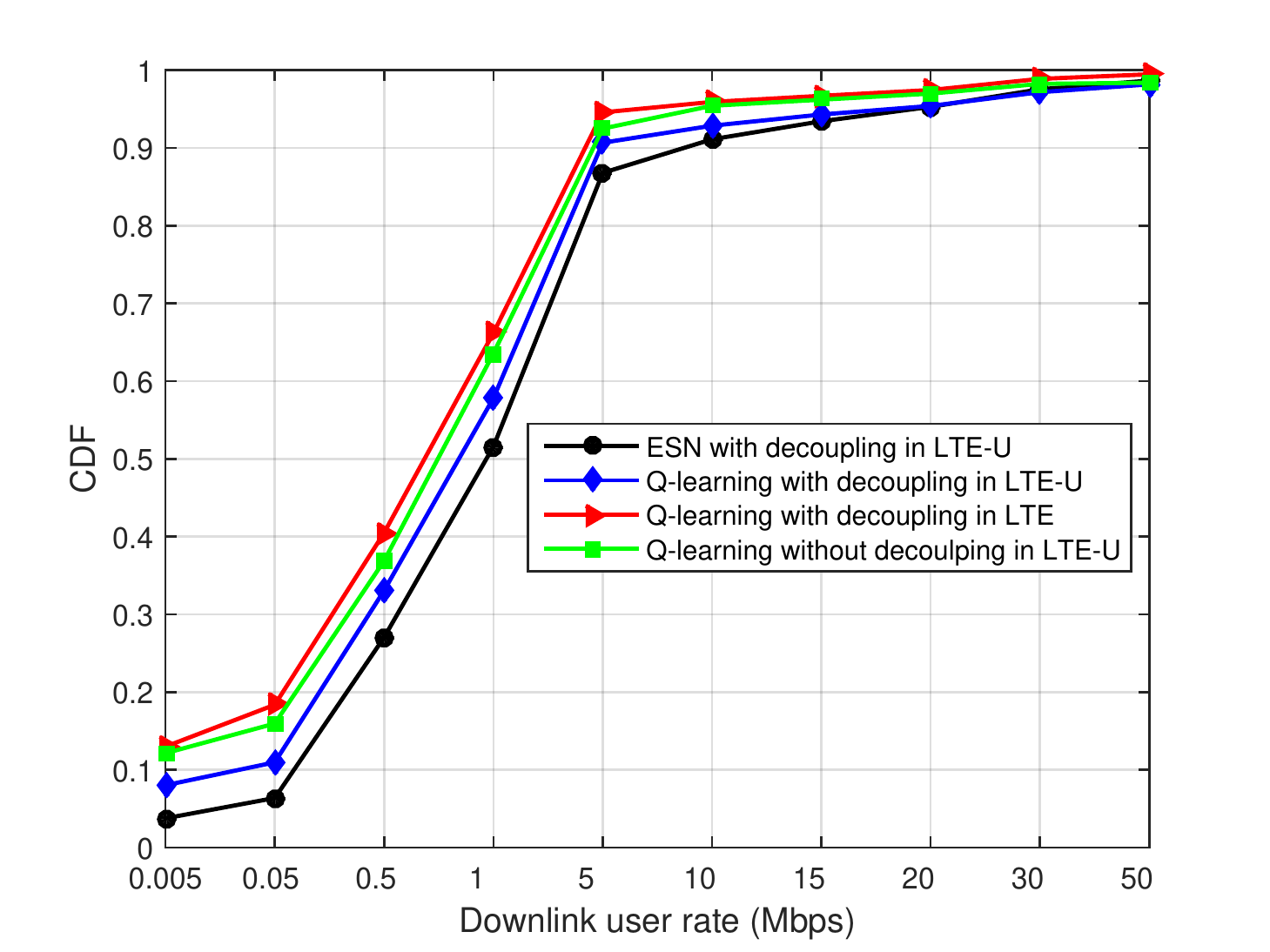}}
  \vspace{-0.2cm}
 \caption{\label{fig7} {CDFs of the downlink and uplink rates} resulting from the different algorithms ($U=40, N_s=5, R_w=4$ Mbps).}
 \vspace{-1cm}
\end{figure}

Fig. \ref{fig7} shows the cumulative distribution function (CDF) of rate in both the uplink and downlink for all the considered schemes. In Fig. \ref{fig7a}, we can see that, the uplink rates of { $7\%$, $11\%$, $13\%$, and $17\%$ of users resulting from} all the considered {algorithms} are below $0.001$ Mbps. This is due to the fact that, in all algorithms, each SBS has a limited coverage which limits the users' {association possibilities}. Fig. \ref{fig7a} also shows that the proposed approach improves the uplink CDF {of up} to {$57\%$ and $143\%$} gains {at a rate of} 0.001 Mbps compared to {Q-learning with decoupling in LTE-U and Q-learning without decoupling in LTE-U}, respectively. 
Fig. \ref{fig7b} also shows that the proposed approach improves the downlink CDF {of up to $120\%$, $28\%$, and $6\%$ for rates of 0.05, 0.5, and 5 Mbps} compared to {Q-learning with decoupling in LTE-U}. {These gains show that the performance of the proposed algorithm is better than that of Q-learning with decoupling in LTE-U.} This is because the proposed approach uses the estimated expected value of the reward function to choose the optimal allocation scheme that results in {an} optimal reward. 

\begin{figure}
\centering
\vspace{0cm}
\subfigure[]{
\label{fig8a} 
\includegraphics[width=9cm]{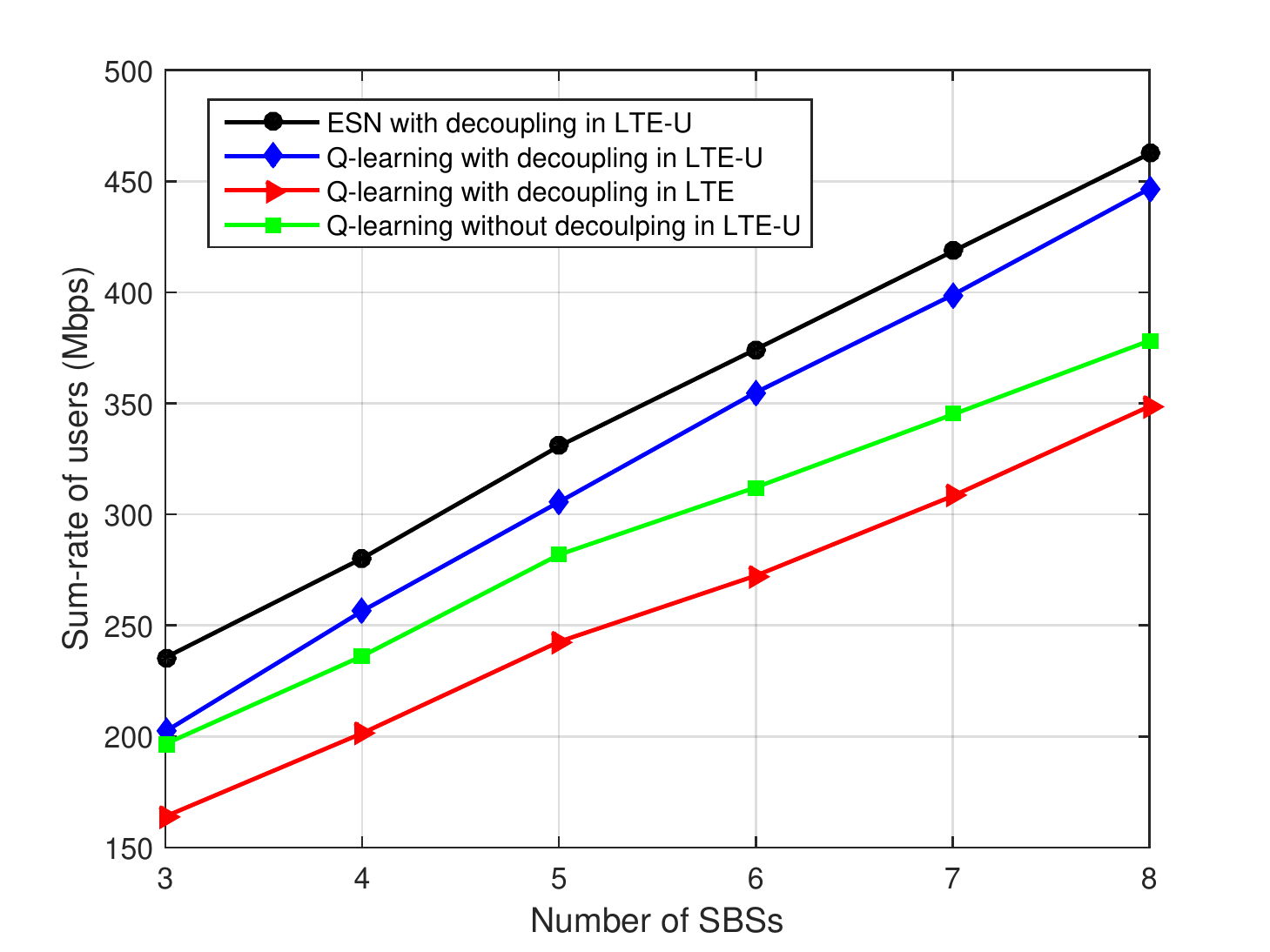}}
\subfigure[]{
\label{fig8b} 
\includegraphics[width=9cm]{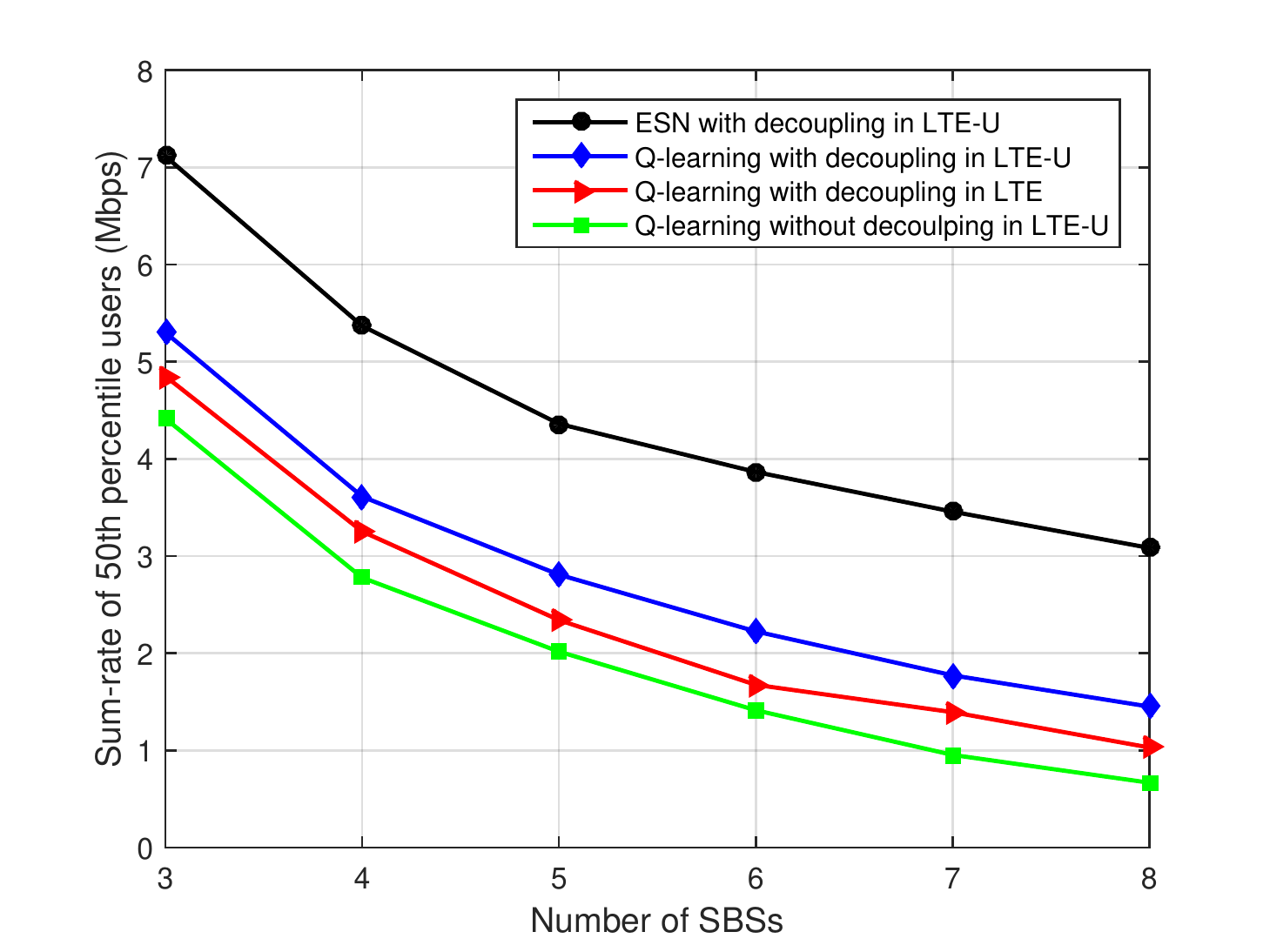}}
  \vspace{-0.2cm}
 \caption{\label{fig8} Sum-rate as the number of SBSs varies ($U=40, R_w=4$ Mbps).}
  \vspace{-1cm}
\end{figure}

In Fig. \ref{fig8}, we show how the total sum-rate varies with the number of SBSs. In Fig. \ref{fig8a}, we can see that, as the number of SBSs increases, all algorithms result in increasing {the} sum-rates because the users have more {SBS choices} and the distances from the SBSs to the users decrease. Fig. \ref{fig8a} also shows that {Q-learning with decoupling in LTE-U} achieves, respectively, up to {$28\%$ and $20\%$} improvements in the sum-rate compared to {Q-learning with decoupling in LTE and Q-learning without decoupling in LTE-U} for the case with $8$ SBSs. {Clearly, the $28\%$ gain is due to the additional use of the unlicensed band and the $20\%$ gain stems from the uplink-downlink decoupling.} 
However, Fig. \ref{fig8b} shows that the rates of the $50$th percentile of users decreases as the number of SBSs increases. This is due to the fact that, in our simulations, each SBS has a limited coverage area that restricts the access of the users. Thus, as the number of SBSs {increases, the interference of the users associated with the MBS increases which results in the decrease of the $50$th percentile of the user} {sum-rate}. By comparing Fig. \ref{fig8a} with Fig. \ref{fig8b}, we can also see that the proposed approach {yields}, respectively, {$3\%$ and $20\%$ gains in terms of the sum-rate compared to {Q-learning with decoupling in LTE-U and Q-learning without decoupling in LTE-U} for $8$ SBSs. The proposed approach also achieves, respectively, around $167\%$ and $400\%$} gains {in terms of} the sum-rate of the $50$th percentile of the users compared to {Q-learning with decoupling in LTE-U and Q-learning without decoupling in LTE-U for $8$ SBSs}. These gains demonstrate that the proposed algorithm {achieves a} better load balancing compared to {each Q-learning approach}. Moreover, Fig. \ref{fig8a} and Fig. \ref{fig8b} also show that {Q-learning with decoupling in LTE achieves} a higher {sum-rate} of the $50$th percentile of users {and a} lower sum-rate {for} all users {compared to Q-learning without decoupling in LTE-U}. It is obvious that the downlink-uplink decoupling improves the rate of edge users.

\begin{figure}
\centering
\vspace{0cm}
\subfigure[Sum-rates of all users]{
\label{fig9a} 
\includegraphics[width=9cm]{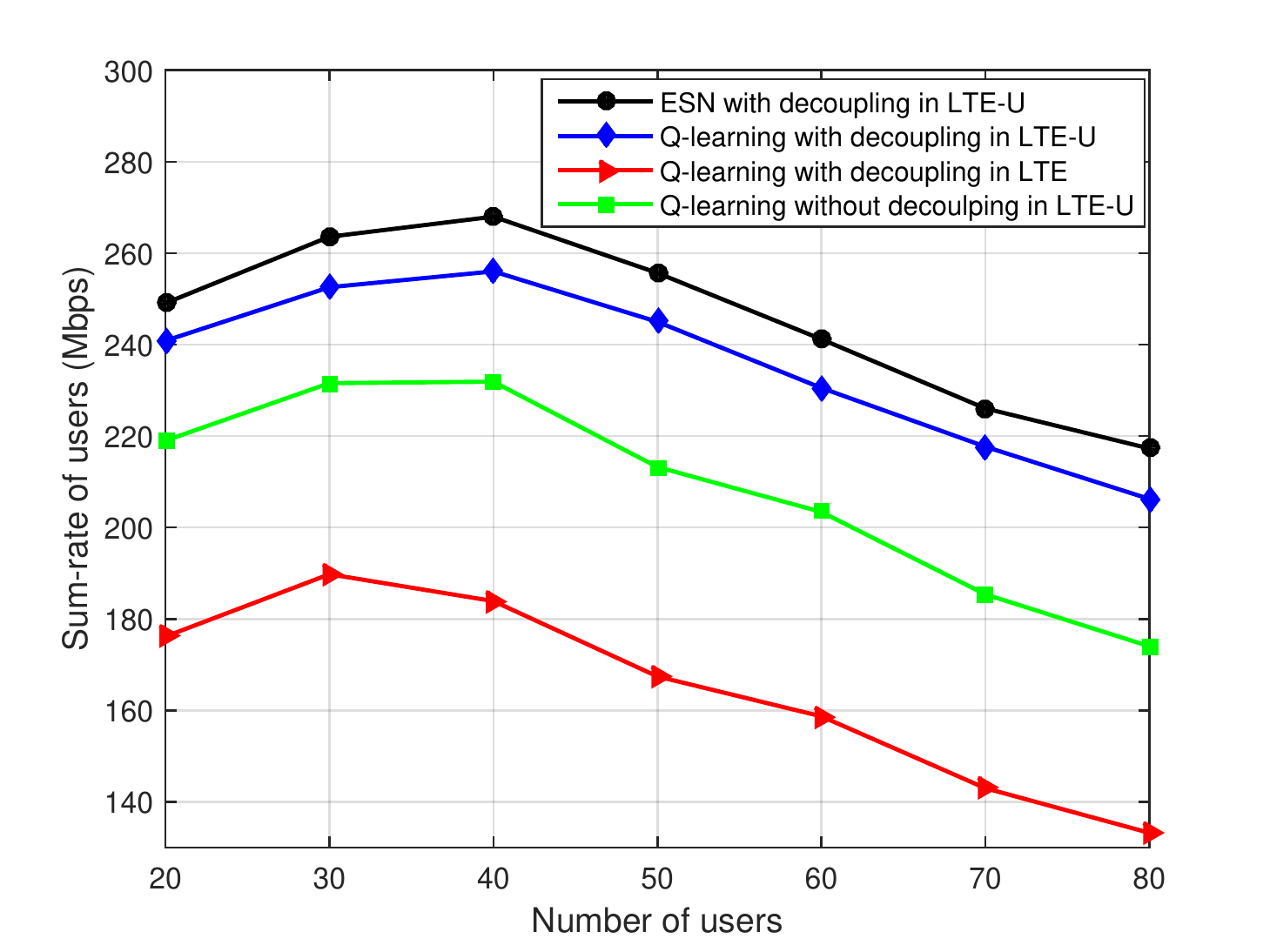}}
\subfigure[Sum-rates of the $50$th percentile of users]{
\label{fig9b} 
\includegraphics[width=9cm]{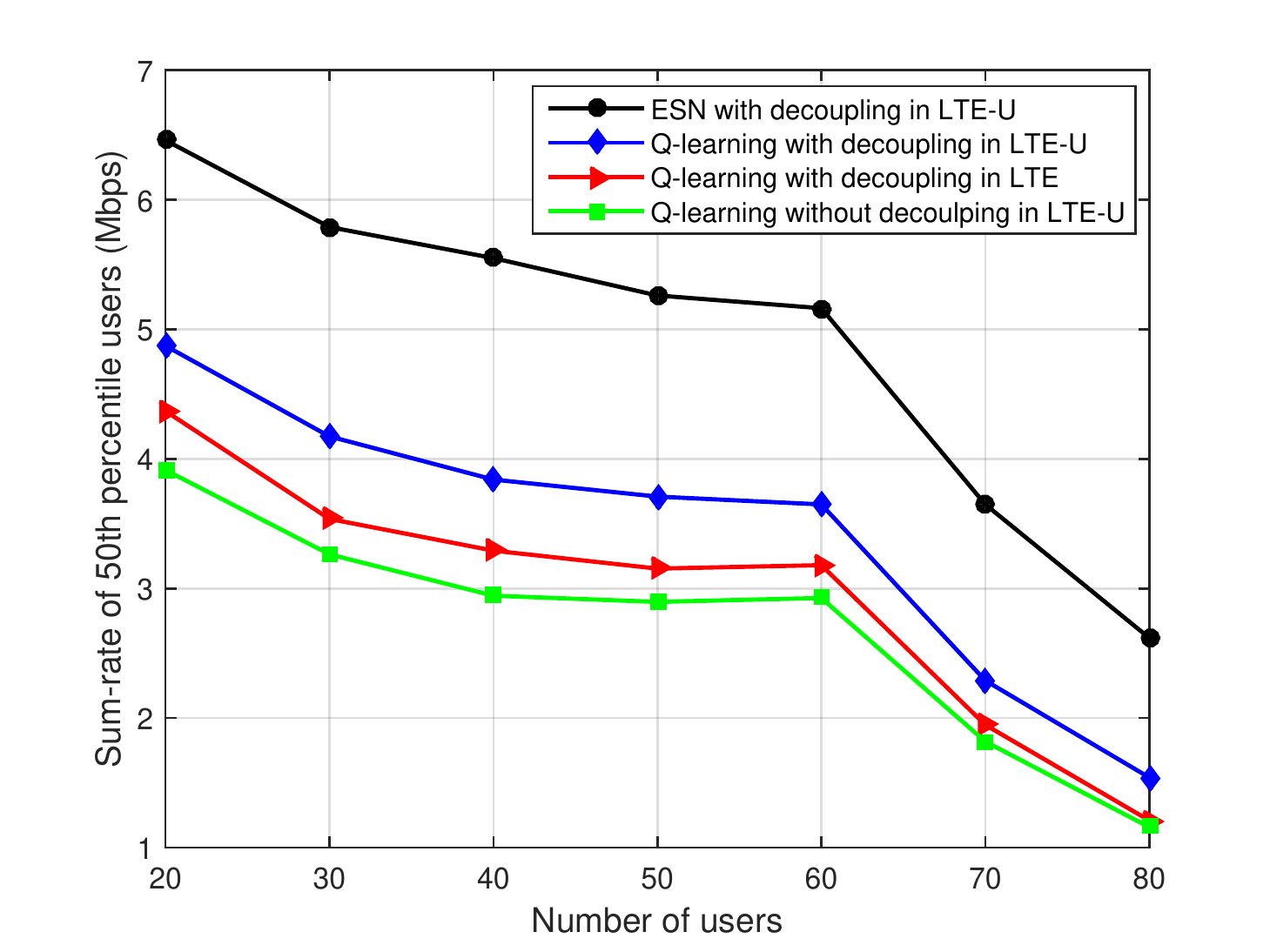}}
  \vspace{-0.2cm}
 \caption{\label{fig9} Downlink and uplink sum-rates of users vs. the number of users ($N_s=4, R_w=4$ Mbps).}
  \vspace{-1cm}
\end{figure}

Fig. \ref{fig9} shows how the total users' sum-rate in both the uplink and downlink changes as the number of {the} users varies. In Fig. \ref{fig9a}, we can see that {the} sum-rate increases then decreases as the number of {the} users increases. That is because each SBS has a relatively small load on the average as the number of {the} users is below $40$. However, as the number of {the} users {goes} above $40$, {the} sum-rate also decreases because each SBS needs to allocate more spectrum to the users who have small SINRs. {From Fig. \ref{fig9b}, we can also see that}, as the number of {the} users increases, the sum-rate of the $50$th percentile of {the} users decreases. {However}, this decrease is much slower for all considered algorithms as the number of {the} users is below $60$. Fig. \ref{fig9b} also shows that {for more than $60$ users}, the {sum-rate} of the $50$th percentile of {the} users for all considered algorithms {decreases} much faster than the case {with less than $60$ users}. This is due to the fact that the {SBSs become overloaded}. Fig. \ref{fig9b} also shows that the proposed algorithm achieves, respectively, up to {$42\%$ and $73\%$} improvements in the sum-rate compared to {Q-learning with decoupling in LTE-U} for the cases with 60 users and 80 users. {This implies that,} by using ESN, each BS can learn and decide on the spectrum allocation scheme better than Q-learning {while reaching a} mixed strategy NE. {Moreover, in Fig. \ref{fig9b}, we can also see that the deviation between Q-learning with decoupling in LTE and Q-learning without decoupling in LTE-U decreases as the number of the users varies.} {The main reason behind this is that, as the BSs become overloaded, some users will not find an appropriate BS to connect to and each BS will not have enough spectrum to allocate to these users.}

\begin{figure}[!t]
  \begin{center}
   \vspace{0cm}
    \includegraphics[width=9cm]{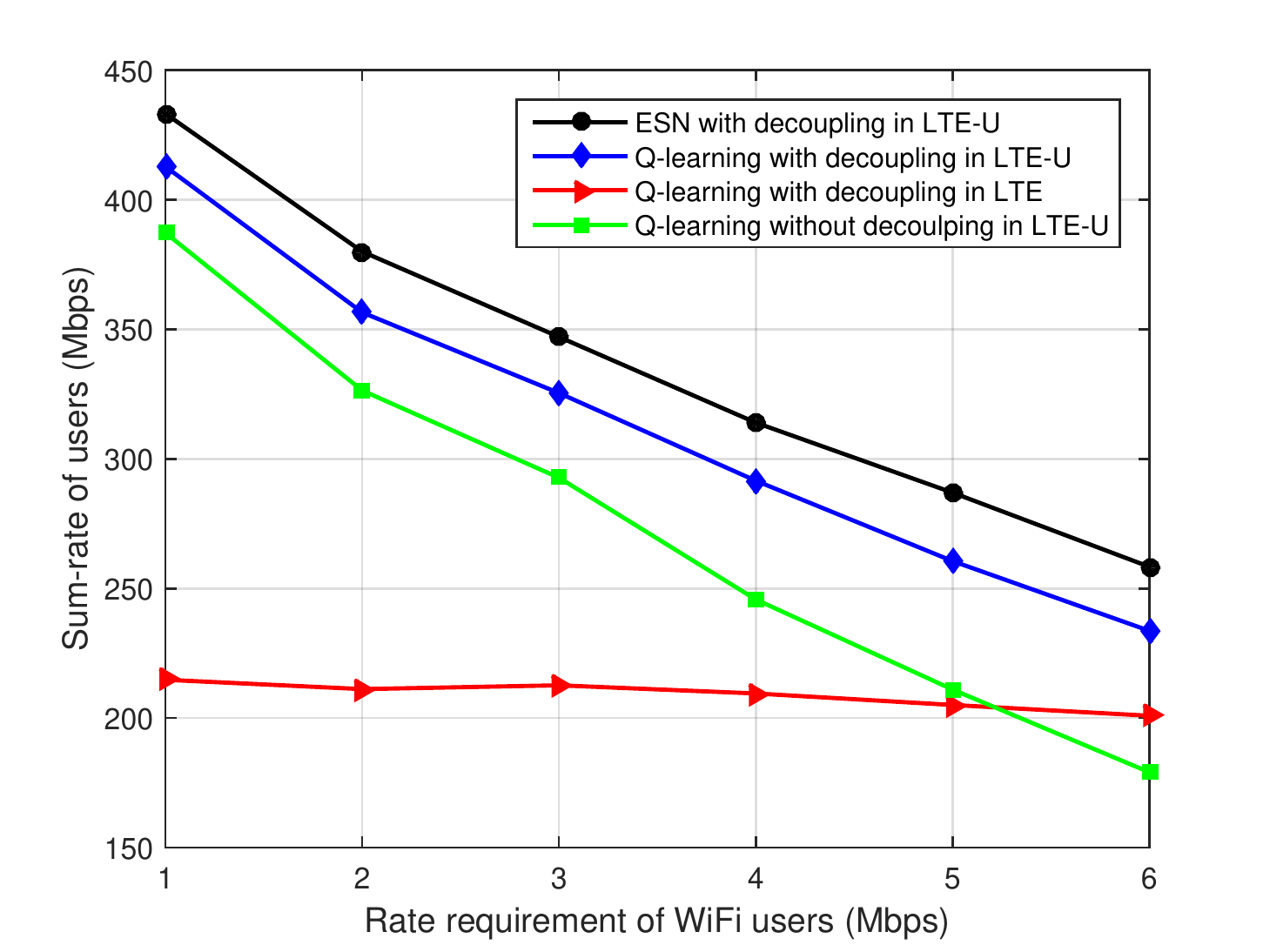}
    \vspace{-0.3cm}
    \caption{\label{fig10} Downlink and uplink sum-rates of users vs. the rate requirement of WiFi users ($U=40, N_s=5$).}
  \end{center}\vspace{-1.2cm}
\end{figure}

\begin{figure}[!t]
  \begin{center}
   \vspace{0cm}
    \includegraphics[width=9cm]{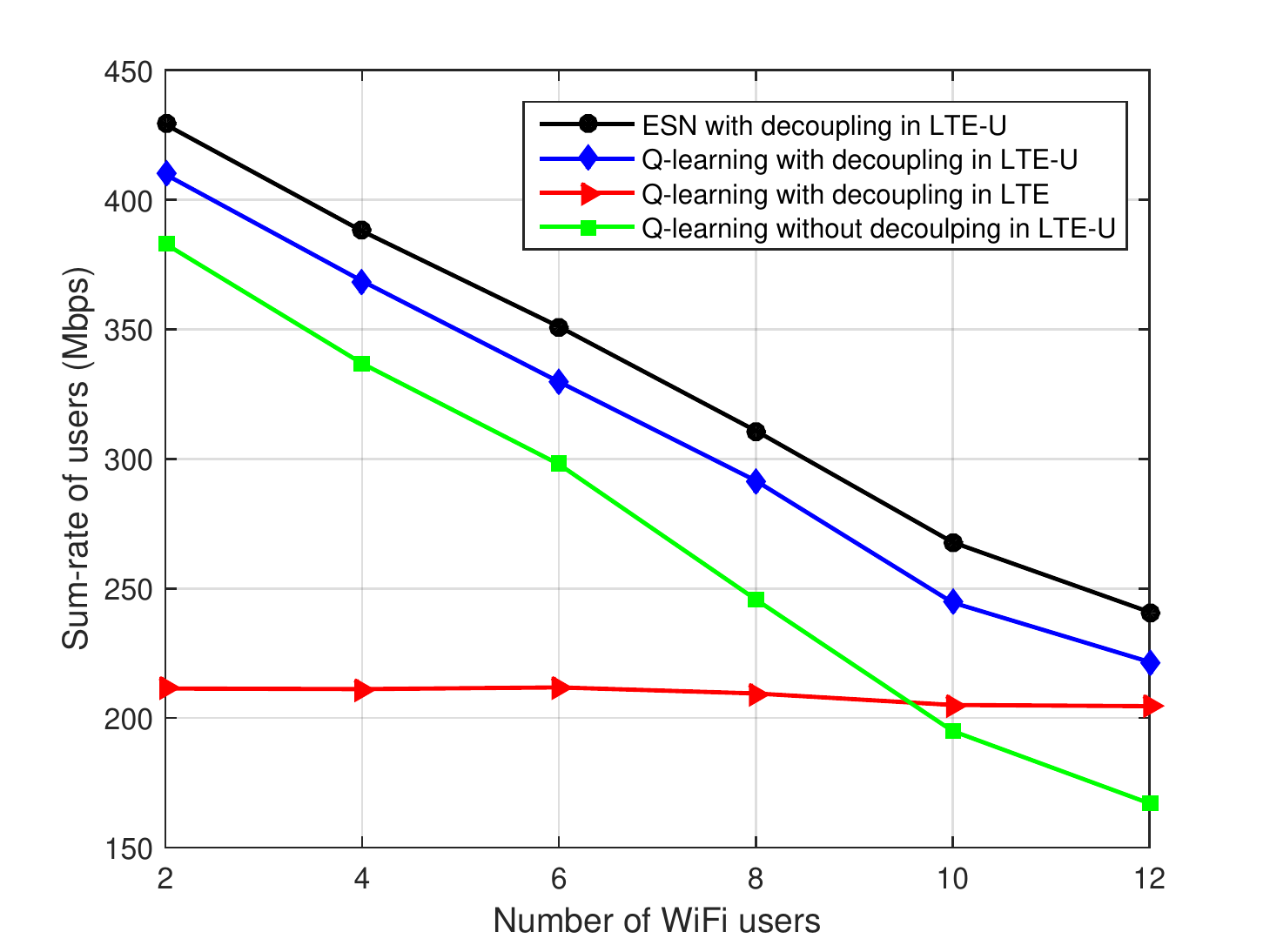}
    \vspace{-0.3cm}
    \caption{\label{fig11} Downlink and uplink sum-rates of users vs. the number of WiFi users ($U=40, N_s=5, R_w=4$ Mbps).}
  \end{center}\vspace{-1.2cm}
\end{figure}    
 
 \begin{figure}
\centering
\vspace{0cm}
\subfigure[The case with convergence of Q-learning]{
\label{fig5a} 
\includegraphics[width=9cm]{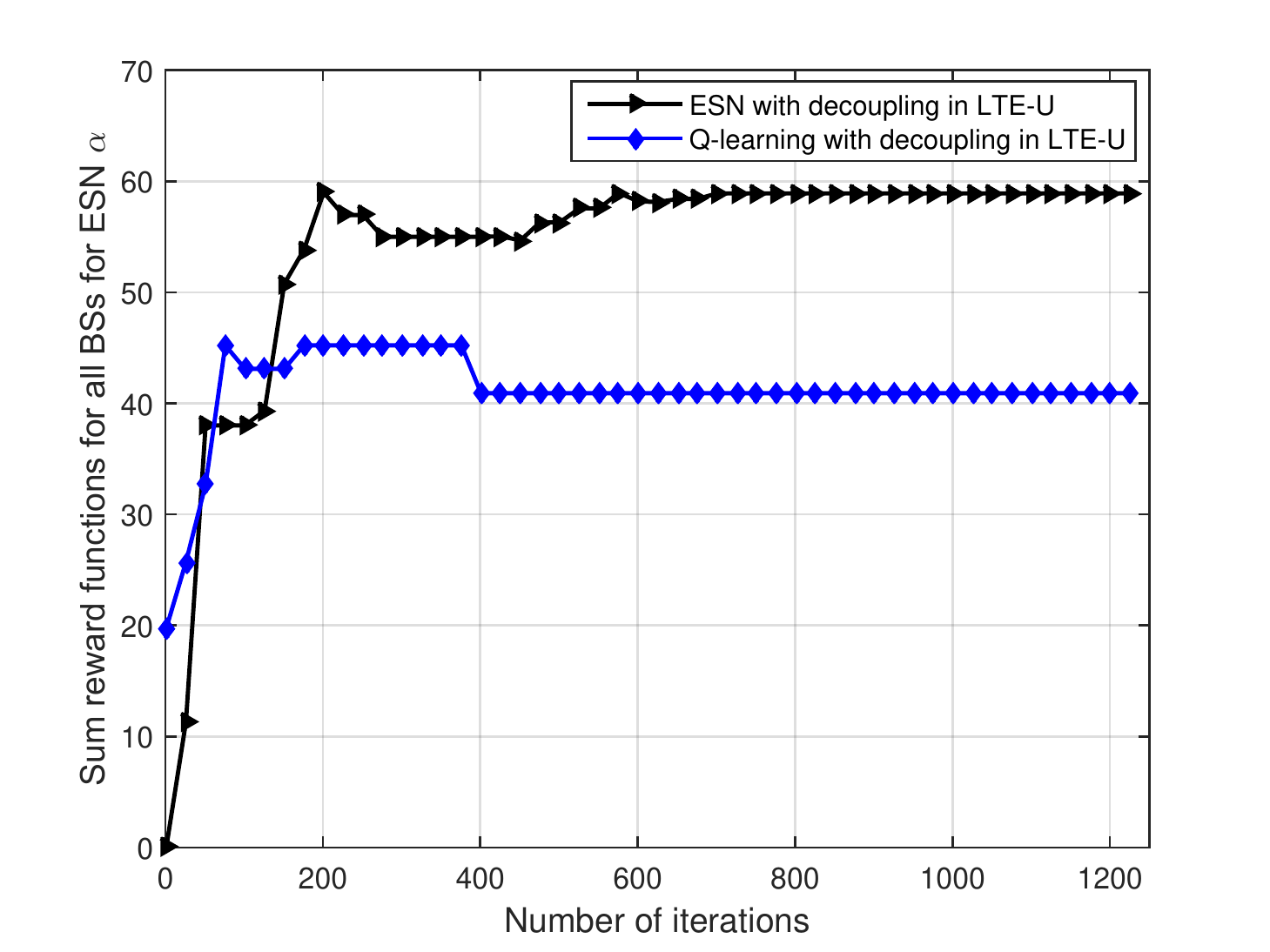}}
\subfigure[The case with divergence of Q-learning]{
\label{fig5b} 
\includegraphics[width=9cm]{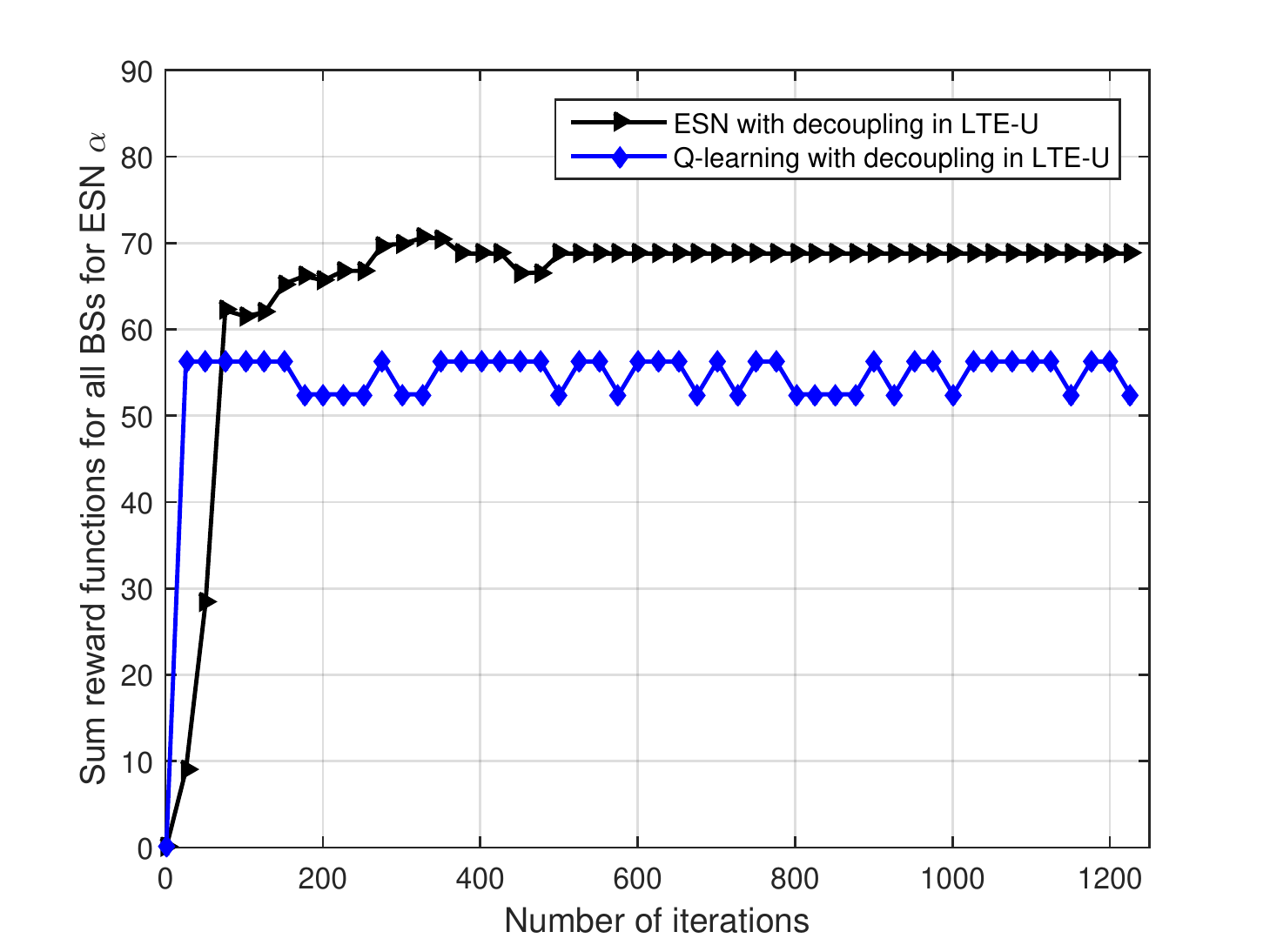}}
  \vspace{-0.2cm}
 \caption{\label{fig5} {Convergence of the proposed algorithm and Q-learning with decoupling in LTE-U} ($N_s=5, U=40, R_w=4$ Mbps).}
  \vspace{-1cm}
\end{figure}
   
Fig. \ref{fig10} shows {how the} total users' sum-rate in both the uplink and downlink changes as the rate requirement of {the} WiFi users varies. In Fig. \ref{fig10}, we can see that {the sum-rates resulting from} all considered algorithms other than {Q-learning with decoupling in LTE} decrease as the rate requirement of {the} WiFi users increases. Fig. \ref{fig10} also shows that {Q-learning without decoupling in LTE-U yields a} lower sum-rate compared to {Q-learning with decoupling in LTE when} the rate requirement of {the} WiFi users is above $5$ Mbps. {These are due to the fact that the fraction of the time slots $L$ on the unlicensed band that is allocated to the LTE-U network decreases as the rate requirement of the WiFi users increases.} {Fig. \ref{fig11} shows how the total users' sum-rate for both the uplink and downlink changes as the number of {the} WiFi users varies. From Fig. \ref{fig11}, we can see that the sum-rate of {the} LTE-U users decreases as the number of the WiFi users increases. This is due to the fact that, as the number of the WiFi user increases, the fraction of time slots on the unlicensed band that is allocated to the LTE-U network decreases.}


\begin{figure}[!t]
  \begin{center}
   \vspace{0cm}
    \includegraphics[width=9cm]{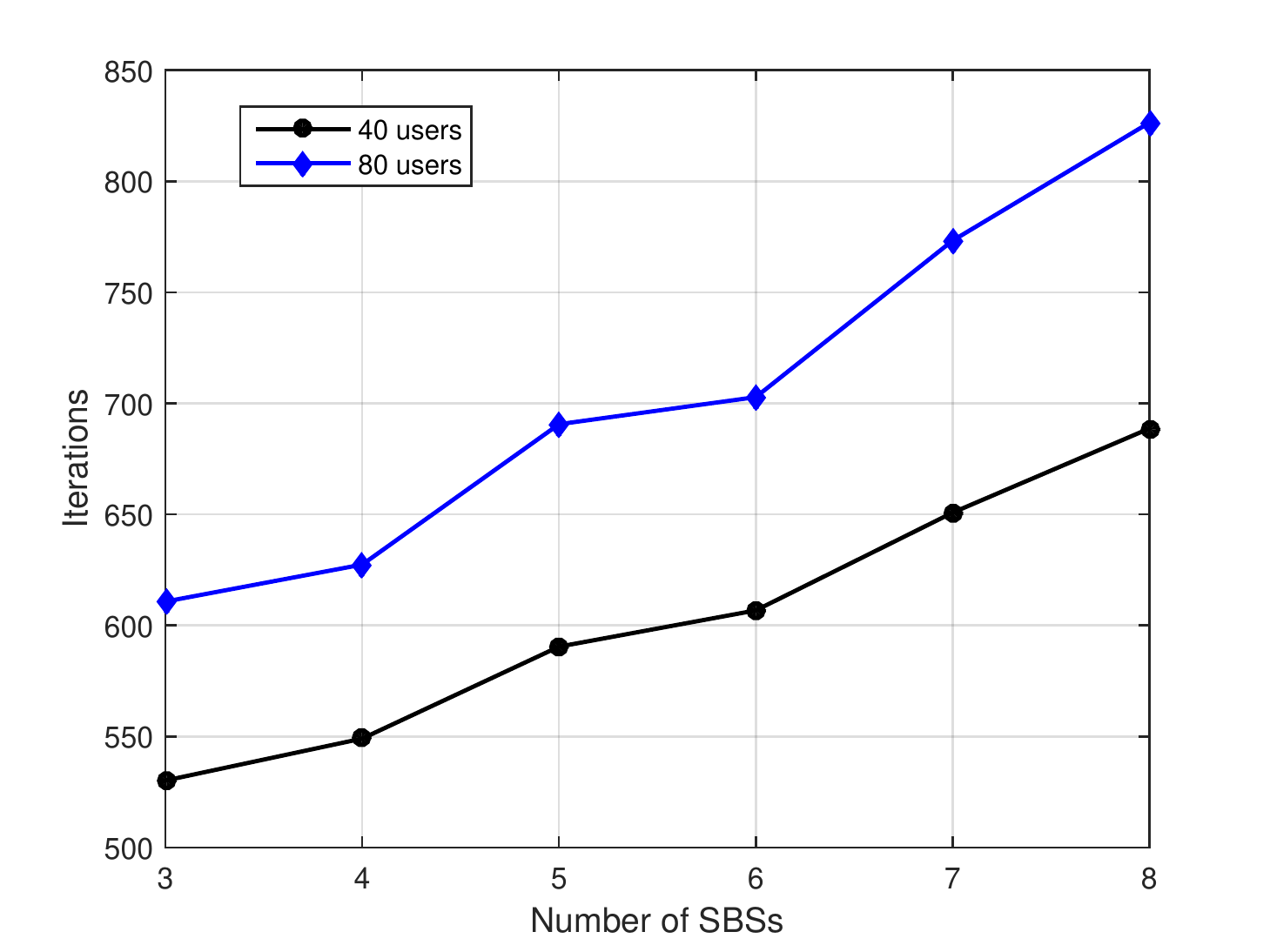}
    \vspace{-0.3cm}
    \caption{\label{fig6}The convergence time of the proposed algorithm as a function of the number of SBSs ($R_w=4$ Mbps).}
  \end{center}\vspace{-1.2cm}
\end{figure}

Fig. \ref{fig5} shows the number of iterations needed till convergence for both the proposed approach and {Q-learning with decoupling in LTE-U}. In this figure, we can see that, as time elapses, the total value of {the} reward {functions} increase until convergence to their final values. In Fig. \ref{fig5a}, we can see that the proposed approach needs {600} iterations to reach convergence. {Moreover, the proposed algorithm exhibits an acceptable increase in terms of the number of iterations needed to converge to a mixed strategy NE compared to Q-learning. This stems from the fact that Q-learning must explore} the entire {information} of all users and BSs to update the Q-table, but the proposed algorithm only needs the action {information} of {the} BSs to update output matrix.     
However, Fig. \ref{fig5b}  shows that the proposed algorithm {eventually converges to an equilibrium, unlike the Q-learning algorithm which oscillates since the action update strategy in Q-learning does not necessarily maximize the expected reward and, as such, Q-learning may not converge to an equilibrium.}



In Fig. \ref{fig6}, we show the convergence time of the proposed approach as the number of SBSs varies for 40 and 80 users. In this figure, we can see that, as the network size increases, the average number of iterations {needed until convergence} increases. Fig. \ref{fig6} also shows that reducing the number of {the} users leads to a faster convergence time. Although {the} users are not players in the game, they affect {the spectrum allocation choices} of each BS. As the number of {the} users increases, the spectrum allocation action for each BS increases, and, thus, a longer convergence time is observed. {In Fig. \ref{fig6}, we can see that the proposed ESN-based algorithm requires less than 800 iterations for the case with $8$ SBSs and $80$ users. Here, we note that such a convergence is significantly faster than existing related learning algorithms such as in \cite{14} and \cite{21} where the number of iterations needed for convergence exceeds 1000, thus further demonstrating that the proposed algorithm can be implemented with an acceptable number of iterations.}

\section{CONCLUSION}
In this paper, we have developed a novel resource allocation framework for optimizing the use of uplink-downlink decoupling in an LTE-U system. We have formulated the problem as a noncooperative game between the BSs that seeks to maximize the total uplink and downlink rates while balancing the load among one another. To solve this game, we have developed a novel algorithm based on the machine learning tools of echo state networks. The proposed algorithm enables each BS to decide on its spectrum allocation scheme autonomously with limited information on the network state. Simulation results have shown that the proposed approach yields significant performance gains in terms of rate and load balancing compared to conventional approaches. Moreover, the results have also shown that the use of ESN can significantly reduce the information exchange for the wireless networks.   

\bibliographystyle{IEEEbib}
\bibliography{references}

\end{document}